 \documentclass[10pt,twocolumn,twoside]{IEEEtran}

\usepackage{cite}            
\usepackage{graphicx}          
\usepackage[cmex10]{amsmath} 
\usepackage{amsfonts} 
\usepackage{amsthm}
\usepackage{mathtools} 
\usepackage{color} 
\usepackage{subcaption}
\usepackage{stfloats} 
\usepackage{tikz}
\usetikzlibrary{arrows,
decorations.pathmorphing,decorations.pathreplacing, backgrounds, positioning,
fit, shapes.geometric}

\tikzstyle{tableEnc} = [minimum size=3mm, thick, node
distance=5mm, circle,draw=black,fill=white,thick, inner sep=0.3pt, outer
sep=1pt]
\tikzstyle{tableEl} = [minimum size=3mm, none, node
distance=5mm, rectangle,draw=black,fill=white,thick, inner sep=0.3pt, outer
sep=1pt]

\theoremstyle{plain}
\newtheorem{thm}{Theorem}
\newtheorem{lem}[thm]{Lemma}
\newtheorem{cor}[thm]{Corollary}
\theoremstyle{definition}

\newtheorem*{problem}{Problem statement}
\begin{document}  

\tikzstyle{genGraphNode} = [minimum size=3mm, thick, node
distance=5mm] 
\tikzset{
    >=stealth', bend angle=10 }
\tikzstyle{normalNode} = [genGraphNode, circle,draw=black,fill=black!10,thick]
\tikzstyle{controllingNode} = [genGraphNode,
circle,draw=black!20,fill=black!60,very thick, text=white, font=\bfseries] 
\tikzstyle{observingNode} =[genGraphNode, circle,draw=black!50,fill=white,thick]
\tikzstyle{placeholder}=[genGraphNode, circle]

\tikzstyle{block} = [draw, fill=white!20, rectangle, 
    minimum height=2em, minimum width=3em] 
\tikzstyle{gain} = [draw,shape=circle, minimum height=2em, minimum
width=2em, inner sep=1pt] 
\tikzstyle{gainLeft} = [gain]
\tikzstyle{gainRight} = [gain, shape border rotate=180]
\tikzstyle{gainDown} = [gain, shape border rotate=-90]
\tikzstyle{gainUp} = [gain, shape border rotate=90]
\tikzstyle{sum} = [draw, fill=white!20, circle,inner sep=1pt]
\tikzstyle{input} = [coordinate]
\tikzstyle{output} = [coordinate]
\tikzstyle{pinstyle} = [pin edge={to-,thin,black}] 
 

\newcommand{\numVeh}{N} 
\newcommand{\pos}{y}
\newcommand{\transPos}{\bar{\pos}}
\newcommand{\err}{\tilde{e}}
\newcommand{\totalInp}{u}
\newcommand{\inp}{r}
\newcommand{\transInp}{\bar{\inp}}
\newcommand{\rearWeight}{\epsilon}
\newcommand{\inpGen}{w}
\newcommand{\outGen}{z}
\newcommand{\disturbance}{d}
\newcommand{\outDist}{{\disturbance_{\text{out}}}}
\newcommand{\inDist}{{\disturbance_{\text{in}}}}
\newcommand{\dreference}{d_{\text{ref}}}

\newcommand{\contTf}{R}
\newcommand{\vehTf}{G}
\newcommand{\openLoop}{M}
\newcommand{\openLoopPart}{\openLoop_s}
\newcommand{\vehNum}{b}
\newcommand{\vehDen}{a}
\newcommand{\contNum}{q}
\newcommand{\contDen}{p}
\newcommand{\olnum}{\phi}
\newcommand{\olden}{\psi} 
\newcommand{\inputCoef}{\rho}
\newcommand{\diagBlock}{T}
\newcommand{\tranFun}{T}
\newcommand{\tranFunFromTo}[2]{T_{#1 #2}}
\newcommand{\tranFunCo}{{\tranFunFromTo{\contNode}{\obsvNode}}}
\newcommand{\tranFunZ}{Z}
\newcommand{\numTf}{n}
\newcommand{\denTf}{d}
\newcommand{\numPol}{h}
\newcommand{\denPol}{g}
\newcommand{\prodPartNum}{\tau}
\newcommand{\oldenord}{\mu}
\newcommand{\olnumord}{\eta}
\newcommand{\relOrd}{\chi}
\newcommand{\orderQ}{N_n} 
\newcommand{\tranFunZeroPol}{\numTf_\Delta}
\newcommand{\tranFunDistOut}[2]{\tranFun_{\text{out},#1 #2}}
\newcommand{\tranFunDistOutCo}{\tranFunDistOut{\contNode}{\obsvNode}}
\newcommand{\tranFunDistIn}[2]{\tranFun_{\text{in},#1 #2}}
\newcommand{\tranFunFeedbackPart}{S}
\newcommand{\tranFunFeedbackPartFromTo}[2]{{\tranFunFeedbackPart}_{#1 #2}}
\newcommand{\tranFunFeedbackPartCo}{\tranFunFeedbackPartFromTo{\contNode}{\obsvNode}}
\newcommand{\tranFunGenInpState}[2]{\tranFun_{\text{\inpGen\outGen}, #1 #2}}

\newcommand{\stateVar}{x}
\newcommand{\stateVarDiag}{\hat{\stateVar}}
\newcommand{\matA}{A}
\newcommand{\matB}{B}
\newcommand{\matC}{C}
\newcommand{\matD}{D}
\newcommand{\contMat}{\mathcal{C}}
\newcommand{\obsvMat}{\mathcal{O}}
\newcommand{\forestMat}{Q}
\newcommand{\lapl}{L} 
\newcommand{\spatEig}{\lambda} 
\newcommand{\spatEigZ}{\gamma} 
\newcommand{\redLapl}{\bar{\lapl}}
\newcommand{\charPolLapl}{p_\lapl}
\newcommand{\charPolLaplCoef}{g}
\newcommand{\charPolNum}{\numPol} 

\newcommand{\charPolNumCoefMod}{\mu} 
\newcommand{\charPolNumCoefSingle}{h}
\newcommand{\charPolNumCoef}{\bar{\charPolNumCoefSingle}} 
\newcommand{\prodPartNumCoef}[2]{\bar{\tau}^{#1}_{#2}}
\newcommand{\eigVect}{v}
\newcommand{\matEigVect}{V}
\newcommand{\matJ}{\Lambda}
\newcommand{\canonVect}{e}
\newcommand{\idMat}{I}
\newcommand{\oneVect}{\mathbf{1}}

\newcommand{\graph}{\mathcal{G}}
\newcommand{\vertexSet}{\mathcal{V}}
\newcommand{\edgeSet}{\mathcal{E}}
\newcommand{\edge}{\epsilon}
\newcommand{\vertex}{\nu}
\newcommand{\adjMat}{A}
\newcommand{\degreeMat}{D}
\newcommand{\neighborSet}{\mathcal{N}}
\newcommand{\path}[2]{\pi_{#1 #2}}
\newcommand{\pathLength}{{l}}
\newcommand{\pathWeight}{\vartheta}
\newcommand{\pathWeightCo}{\pathWeight_{\contNode \obsvNode}}
\newcommand{\distance}{{\delta}}
\newcommand{\distanceCO}{\distance_{\contNode \obsvNode}}
\newcommand{\tree}{\mathcal{T}}
\newcommand{\forest}{\bar{F}}
\newcommand{\forestSet}{\mathcal{F}}

\newcommand{\obsvNode}{o}
\newcommand{\contNode}{c}
\newcommand{\contNodeSet}{\mathcal{S}_\contNode}
\newcommand{\contNodeNum}{N_c}             

\title{Transfer functions in consensus systems with
higher-order dynamics and external inputs}   

 \author{Ivo~Herman, Dan~Martinec and Michael Sebek
\thanks{All authors are with the Faculty of Electrical Engineering, Czech
Technical University in Prague, Department of Control Engineering, Karlovo
namesti 13, 121 35 Prague, Czech Republic. 
E-mail address:
\{ivo.herman, martinec.dan, sebekm1\}@fel.cvut.cz }
\thanks{The research
was supported by the Czech Science Foundation within the project GACR
13-06894S (I.~H.). 
}
\thanks{Manuscript received January 23, 2015, revised August, 2015}} 
                                          
\markboth{IEEE Transactions on Control of Network Systems}%
{Herman \MakeLowercase{\textit{et al.}}: Transfer functions in consensus
systems}

\maketitle

\begin{abstract}                          
This paper considers transfer functions
in consensus systems where agents have identical SISO dynamics
of arbitrary order. The interconnecting structure is a directed
graph.
The transfer functions for various inputs and outputs are
presented in simple product forms with a similar structure of the numerator and
the denominator. This structure combines the network properties and the agent model in an explicit way. 
The link between a higher-order and a single-integrator
dynamics is shown and the polynomials of the transfer function in the
single-integrator system are related to the graph properties. These properties also allow to generalize a
result on the minimal
dimension of the controllable subspace to the directed graphs.
\end{abstract} 

\section{Introduction} 
Distributed control has become a very intensive field of research. Numerous
results for control of highway platoons, robot formations or synchronization of
oscillators were published. To the standard consensus problem
an exogenous input can be added, which could capture for instance the effect of
a measurement noise, input or output disturbances \cite{Lin2012a, Fitch2013,
Bamieh2012} or reference values \cite{Herman2013b}.

Much effort has been invested in understanding the behavior of single integrator
systems, especially of the consensus. The results reveal the effect the network
structure has on the overall behavior.  For example, the convergence time is
related to the second smallest eigenvalue
of the Laplacian matrix \cite{Olfati-Saber2007}.
The effect of the network structure on the $\mathcal{H}_2$ and
$\mathcal{H}_\infty$ norms was investigated in \cite{Zelazo2011}.

In many systems (e.g., highway platoons or oscillators), the agent models are
more complicated higher-order systems.  In this case stability is a crucial issue.
 It was shown in the paper \cite{Fax2004a} that the overall
formation of identical agents is stable if and only if it is stable for all
eigenvalues of a Laplacian matrix. One of the approaches for stabilization is
 based on changing the
gains when the graph topology changes \cite{Fradkov2011, Zhang2011}. Also the
concept of passivity guarantees stability (see \cite{Arcak2007, Chopra2006}).

Nevertheless, stability is not sufficient for good transients.
Phenomena not seen in the single integrator dynamics can appear with
higher-order dynamics. Well known is a so called string
stability, which concerns amplification of the disturbance
in a vehicular formation.
Some of the works in this field using the properties of the transfer functions
are \cite{Seiler2004a, Herman2013b, Middleton2010}. The effect of noise on the
rigidity of the formation, known as coherence, was studied in the paper
\cite{Bamieh2012}.

When the effects of inputs are considered, a controllability becomes an issue.
The results on controllability inferred from graph structure are shown in
\cite{Egerstedt2012}. The bounds on the dimension of the
controllable subspace for undirected graphs are provided in \cite{Zhang2014}.

Dynamic behavior and frequency response of a linear system are given by the
poles and zeros of the transfer function from the input to the output. Their
location also determines the response to some reference signal.
The structure of poles of the transfer functions was investigated in
\cite{Wu1995} or \cite{Olfati-Saber2007}. While the location of poles of a
network systems is now well understood, less attention was paid to the
location of zeros.

One of the first papers considering the location of zeros in the consensus based
algorithms was \cite{Briegel2011}, considering a single integrator model of one
agent and a symmetric communication structure. The paper
\cite{Torres2013} extends the results of \cite{Briegel2011} to the directed
graphs and also shows relations of the zeros to the Laplacian matrix.
Transfer functions and their margins in cyclic formations are discussed in
\cite{Yoon2011}. The paper \cite{Torres2014} shows that even a formation with
stable poles can have zeros in the right half-plane. 

In this paper we study transfer functions in a network system where one
agent with a known input acts as a controlling node
and some other agent, output of which is of interest, serves as an observing
node.
All the agents are modelled by SISO systems and they are interconnected over
directed graphs using relative output feedback. We generalize the results of
\cite{Briegel2011} and \cite{Torres2013} to higher-order dynamics and directed graphs. The key results
are: 
\begin{enumerate}
  \item A product form of the transfer function (Theorem
\ref{lem:totTranFun}) showing a similar structure of the poles and zeros. Such
a product form expresses the transfer function as a series connection of
systems with identical structure. The transfer function with
a general input and output consists of two parts: a network part and an
open-loop part (Theorem \ref{thm:genInputState}).
\item A graph theoretical representation of the polynomials in single-integrator
system (Lemma \ref{lem:coefForestWeights}) and the relation of the zeros to
the Laplacian (Theorem \ref{conj:onePath}). If there is only one path between
the controlling and the observing node, then the zeros are obtained from the
Laplacian matrix.
\item The minimal dimension of the controllable subspace is related to the
maximal distance from the controlling node (Theorem \ref{thm:minContSubsp}). 
\end{enumerate}

Since we work with an arbitrary LTI model, the results here do not tell us much
about particular transient properties. The results should rather serve as tools
for analysis in performance assessment in particular graph types. For instance,
the product form of the transfer function allowed us easier analysis of a
scaling of the $\mathcal{H}_\infty$ norm in vehicular platoons 
\cite{Herman2013b}.

 This paper extends our preliminary results in \cite{Herman2014a}.
We add graph theoretic representations of all polynomials and different types of
inputs and outputs are considered. Also a result on the minimal dimension of the
controllable subspace is added.

\subsubsection*{Notation} We
denote matrices with capital letters and a particular element in a matrix $A$ is
denoted as $a_{ij}$. All vectors are column vectors and are denoted with
lowercase letters, the $i$th element of a vector $v$ is $v_i$. Scalars are
denoted by Greek letters. $\idMat$ is an identity matrix and a canonical basis
vector is $\canonVect_i=[0, \ldots,1, \ldots,0]^T$ with 1 on the $i$th position. The symbol $s$ used
in transfer functions denotes the Laplace variable. The polynomials are denoted
by lowercase letters and $g_i$ is the coefficient at $s^i$ in the polynomial
$g(s)$ (the argument $s$ is usually used with polynomials).

\section{Graph theory}  
The network system interconnection (sharing of information) can be viewed as
a~\emph{directed graph}. The graph $\graph$ has a~vertex set
$\vertexSet(\graph)$ and an arc set $\edgeSet(\graph)$. The arc
$\edge(\vertex_j, \vertex_i)$ is oriented, which means that the $i$th agent
receives its information from the $j$th agent. A directed path $\path{i}{j}$ from $i$ to $j$ of length $\pathLength(\path{i}{j})$ is a sequence of vertices and arcs $\vertex_1,
\edge_1, \vertex_2, \edge_2, \ldots \vertex_{\pathLength+1}$, where each vertex
and arc can be used only once. The length (number of arcs) of the shortest path
between $i$ and $j$ is called the distance $\distance_{ij}$ of vertices. 
A cycle is a path with the first and last vertices identical.

An adjacency matrix is defined as $\adjMat=[a_{ij}]$.
Its entries $a_{ij}$ are either zero if there is no arc from $\vertex_j$ to
$\vertex_i$ or a positive number called weight if the arc is present. We also
define the weight of the path as
$\pathWeight(\path{i}{j})=\prod_{{\edge(k, m) \in
\path{i}{j}}} a_{k m}$.
It is the product of weights of all arcs in the path. Similarly, we define the weight of a~subset $\graph^{'}$ of a~graph $\graph$ as 
\begin{equation}
	\pathWeight(\graph^{'})=\;\prod_{\mathclap{\edge(k,m)\in \edgeSet(\graph^{'})}}
	\; a_{k m}.
\end{equation}

\begin{figure}
\centering
	\centering
\begin{subfigure}{0.14\textwidth}
\centering
	\begin{tikzpicture}
	\node[controllingNode] (n1) {1};
	\node[normalNode] (n2) [below=of n1]{2}
		edge [<-, bend left] (n1);
	\node[observingNode] (n3) [below=of n2]{3}
		edge [<-, bend left] (n2)
		edge [->, bend right] (n2);
	\node[normalNode] (n5) [right=of n2]{5}
		edge [->, bend right] (n2);
	\node[normalNode] (n4) [above=of n5]{4}
		edge [->, bend right] (n5);
	\node[normalNode] (n6) [below=of n5]{6}
		edge [<-, bend left] (n3);	
\end{tikzpicture}
	\caption{Orig. graph}
\end{subfigure} 
\begin{subfigure}{0.14\textwidth}
\centering
	\begin{tikzpicture}
	\node[controllingNode] (n1) {1};
	\node[normalNode] (n2) [below=of n1]{2}
		edge [<-, bend left] node[left] {$0.6$} (n1);
	\node[observingNode] (n3) [below=of n2]{3}
		edge [<-, bend left] node[left] {$0.4$} (n2);
	\node[normalNode] (n5) [right=of n2]{5};
	\node[normalNode] (n4) [above=of n5]{4};
	\node[normalNode] (n6) [below=of n5]{6}
		edge [<-, bend left] node[above] {$1.5$} (n3);	
\end{tikzpicture}
	\caption{First forest}
\end{subfigure}
\begin{subfigure}{0.14\textwidth}
\centering
	\begin{tikzpicture}
	\node[controllingNode] (n1) {1};
	\node[normalNode] (n2) [below=of n1]{2}
		edge [<-, bend left] node[left] {$0.6$} (n1);
	\node[observingNode] (n3) [below=of n2]{3}
		edge [<-, bend left] node[left] {$0.4$} (n2);
	\node[normalNode] (n5) [right=of n2]{5};
	\node[normalNode] (n4) [above=of n5]{4}
		edge [->, bend right] node[left] {$0.8$} (n5);
	\node[normalNode] (n6) [below=of n5]{6};
\end{tikzpicture}
	\caption{Second forest}
\end{subfigure}
 \caption{Example of the set $\forestSet_3^{1\rightarrow3}$ of all spanning
 forests with three arcs with a tree diverging from the node $1$ and containing
 $3$.
 The weights of the two spanning forests are: (a)
 $\pathWeight(\forest_{3}^{1\rightarrow 3})_1=0.6 \cdot 0.4 \cdot 1.5=0.36$ and
 (b) $\pathWeight(\forest_{3}^{1\rightarrow 3})_2=0.6 \cdot 0.4 \cdot
 0.8=0.192$. The weight of the set is $\pathWeight(\forestSet_{3}^{1\rightarrow 3})=0.192+0.36=0.552$. }
	\label{fig:forestExamples}  
\end{figure}
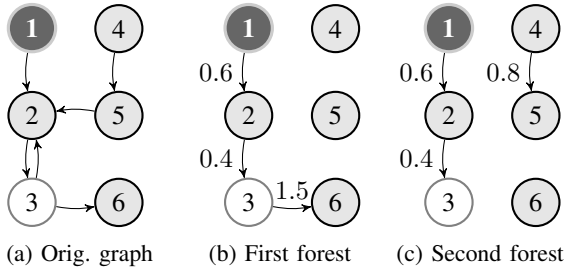

A directed tree is a subset of a graph without directed cycles. A~diverging
directed tree always has a path from one particular node called the
root to each node in the tree.
There is no directed path from the nodes in the diverging tree to the root and
all the nodes except for the root have in-degree one.
A forest $\forest$ is a set of mutually disjoint trees.
A spanning forest is a forest on all vertices of the graph (see
\cite{Chebotarev2014} for an overview of directed trees). A diverging
forest (out-forest) is a forest of diverging trees. 
Following the notation of
\cite{Chebotarev2002} we denote $\forestSet_k^{i\rightarrow j}$ the \emph{set of
all spanning diverging forests} with $k$ arcs.
Such a set must contain a tree with the root $i$ which contains the node $j$.
The weight of this set is
\begin{equation}
	\pathWeight(\forestSet_k^{i\rightarrow j}) =\:\: \sum_{\mathclap{\forest_k^{i
	\rightarrow j} \in \forestSet_k^{i\rightarrow j} }} \:\:
	\pathWeight(\forest_k^{i\rightarrow j}),
\end{equation}
with the sum taken over all spanning forests $\forest_k^{i\rightarrow
j}$ in the set ${\forestSet}_k^{i\rightarrow
j}$. This is illustrated in Fig. \ref{fig:forestExamples}.

Let $Q_{k}$ be a matrix of spanning out-forests of $\graph$ which have $k$
	arcs. The $(i,j)$th element $(q^{k})_{ij}$ of $Q_{k}$ is given as
	\begin{equation}
		(q^{k})_{ij} = \pathWeight(\forestSet_{k}^{j \rightarrow i}).
		\label{eq:weightOfForestCofactor}
	\end{equation}
	It is the weight of the set of all spanning out-forests
	$\forestSet_{k}^{j \rightarrow i}$ with $k$ arcs containing $i$
	and diverging from the root $j$.

 Let us denote $\degreeMat=\text{diag}\big(\text{deg}(\vertex_i)\big)$ the
diagonal matrix of the sums of weights of the arcs incident to the vertex $i$.
Then the Laplacian matrix $\lapl \in \mathbb{R}^{\numVeh \times \numVeh}$ of a
directed graph is defined as
\begin{equation}
	\lapl = \degreeMat-\adjMat. \label{eq:laplacianGen}
\end{equation}

We denote the eigenvalues of the Laplacian as $\spatEig_i, \, i=1, \ldots, N$.
All the eigenvalues have positive real part and there is always a zero
eigenvalue of the Laplacian, i.e., $\spatEig_1 = 0$ with the corresponding
eigenvector $\oneVect$ of all ones, i.e., $\lapl
  \oneVect = 0$.

In the paper we will use a version of Lemma 3.1 in \cite{Briegel2011}. Here
we provide a different proof, as the original proof is valid only for commuting
 matrices and  unweighted
 graphs.
\begin{lem} 
For the elements of the powers of Laplacian holds
	\begin{equation}
		(-\lapl^m)_{ij} = \left\{  \begin{matrix}
			0, & \text{ for } m < \distance_{ji} \\
			\pathWeight(\forestSet_{k}^{j \rightarrow i}), & \text{ for } m =
			\distance_{ji}
		\end{matrix}
		\right.,
	\end{equation} 
	\label{lem:numwalks}
\end{lem}
\begin{proof}
	We will use the result \cite[Proposition 8]{Chebotarev2002}, which shows
	\begin{equation}
		(-\lapl)^m = \sum_{k=0}^m \alpha_k Q_{m-k},
		\label{eq:powersOfLaplacian}
	\end{equation}
	with $\alpha_k \in \mathbb{R}$
	being a constant.
	Since $(q^{m-k})_{ij}$ is the weight of $\forestSet_{m-k}^{j
	\rightarrow i}$, the minimal number of arcs for any forest in the set to exist
	is the distance $\distance_{ji}$ from the node $i$ to the node $i$. Hence, for
	$m<\distance_{ji}$, $(i,j)$th element of all $Q_{m-k}$ is zero and therefore $(-\lapl^{m})_{ij}$
	is also zero. For $m=\distance_{ji}$  the element $(-\lapl^m)_{ij}$ is the sum
	of the weights of all shortest paths.
\end{proof}

\section{System model}
We consider a network system consisting of $\numVeh$ identical agents which
exchange information about their outputs (either using a communication or
measurements).
All are modelled as SISO systems, where dynamic controllers are used. Each agent
is governed locally, therefore no central controller is used.

The plant model
$\vehTf(s)$ (the model of an agent without the controller) is given as
a~transfer function of arbitrary order
	$\vehTf(s) = \frac{\vehNum(s)}{\vehDen(s)}.$
The output of the $i$th plant is denoted as
$\pos_i$.
The plant model is driven by the output of the dynamic controller $\contTf(s)$.
The controller is generally given as a transfer function
	$\contTf(s) = \frac{\contNum(s)}{\contDen(s)}.$
The input to the controller is $\totalInp_i$. As the plant and the controller
are connected in series, the agent model is described by the \emph{scalar} open-loop transfer function
\begin{equation}
	\openLoop(s) = \vehTf(s)
	\contTf(s)=\frac{\vehNum(s)\contNum(s)}{\vehDen(s)\contDen(s)}.
	\label{eq:openLoop}
\end{equation}
The relative degree (the difference between the degree $\oldenord$ of the
denominator and the degree $\olnumord$ of the numerator) of $\openLoop(s)$ is
denoted as $\relOrd=\oldenord - \olnumord$.

The neighbor of an agent $i$ is
defined as an agent $j$ from which the agent $i$ can obtain information about
its output, that is, there exists an arc $\edge(\vertex_j, \vertex_i)$ in the
graph $\graph$. The
relative error of the $i$th agent is defined as
	$\err_i	=  \sum_{j \in \neighborSet(i)} (\pos_j-\pos_i),$
where $\neighborSet(i)$ denotes the set of neighbors of the $i$th agent. 

Apart from the relative error $\err_i$, an exogenous input $\inp_i$ can be
acting at the input of the controller.
The total input to the controller thus is
\begin{equation}
	\totalInp_i	= \err_i + \inp_i = \Big(\displaystyle{\sum_{j \in
	\neighborSet(i)}} (\pos_j-\pos_i)\Big) + \inp_i,
\end{equation}
The input $\inp_i$ can be, for instance, the sum of reference values or some
other external signal such as error in measurement, disturbance etc. We treat
$\inp_i$ as a general signal.
 
\subsection{Problem statement}
The stacked vector of all inputs to the open loops is
\begin{equation} 
	\totalInp(s) = -\lapl \pos(s) + \inp(s), \label{eq:formationInput}
\end{equation}
with $\totalInp=[\totalInp_1, \ldots, \totalInp_\numVeh]^T$, $\pos=[\pos_1,
\ldots, \pos_\numVeh]^T$ and $\inp = [\inp_1, \ldots, \inp_\numVeh]^T$. The
matrix $\lapl$ is the graph Laplacian in (\ref{eq:laplacianGen}).

Now we can write the model of the overall formation as
\begin{IEEEeqnarray}{rCl}
	\pos(s) = \openLoop(s) \totalInp(s) = \openLoop(s) \left[ -\lapl \pos(s) +
	\inp(s) \right]. \label{eq:overallSyst}
\end{IEEEeqnarray}

We are interested in how an exogenous input acting at one selected agent affects
the output of another agent. We assume that there is only one input
$\inp_\contNode$, acting at the input of the agent with index $\contNode$. That is, the input
vector equals $\inp = [0,\ldots, 0,\inp_\contNode,
0,\ldots,0]^T = \canonVect_\contNode \inp_\contNode$. We will call the agent
with index $\contNode$ a \emph{controlling agent}. 

The output of interest is the output $\pos_\obsvNode$ of the
agent with index $\obsvNode$, i. e. the output vector is $\pos = [0, \ldots,
0, \pos_\obsvNode, 0, \ldots, 0]^T = \canonVect_\obsvNode \pos_\obsvNode$. We
call the agent with index $\obsvNode$ an \emph{observing agent}. 
The indices $\contNode$ and $\obsvNode$ can be arbitrary. We will use the statement
``from $\contNode$ to $\obsvNode$'' with the meaning of ``from the input
$\inp_\contNode$ acting at the agent $\contNode$ to the output $\pos_\obsvNode$
of the agent $\obsvNode$''.

Define a transfer function $\tranFunCo(s)$ as
\begin{equation}
	\tranFunCo(s)=\frac{\pos_\obsvNode(s)}{\inp_\contNode(s)}.
	\label{eq:tranFunDef}
\end{equation}

\begin{problem}
	Consider the transfer function $\tranFunCo(s)$ for a network of SISO agents
	connected by a directed graph.
	We study the structure of $\tranFunCo(s)$ and analyze how does
	$\tranFunCo(s)$ depend on the open loop model $\openLoop(s)$, the choice of agents $\contNode$ and $\obsvNode$ and the
	interconnection Laplacian $\lapl$.
\end{problem}

\subsection{Block diagonalization}

We can block diagonalize the system (\ref{eq:overallSyst}) using the transformation
$\pos = \matEigVect \transPos$. The matrix $\matEigVect=[\eigVect_{ij}]$ is a matrix of
(generalized) eigenvectors of the Laplacian, i. e. $\lapl \matEigVect = \matJ \matEigVect$
with $\matJ$ being the Jordan form of $\lapl$. With such a transform, the model
has a form
\begin{IEEEeqnarray}{rCl}
	\matEigVect \transPos(s) &=& \openLoop(s)\left[-\lapl \matEigVect \transPos(s)
	+ \inp(s) \right].
\end{IEEEeqnarray}
Separating $\transPos$ on the left-hand side using $\matJ =
\matEigVect^{-1} \lapl \matEigVect$ yields
\begin{IEEEeqnarray}{rCl}
	\left[ \idMat + \matJ \openLoop(s) \right] \transPos(s) &=&
	\openLoop(s) \matEigVect^{-1} \inp(s). \label{eq:blockDiagSyst} 
\end{IEEEeqnarray}
We can
define the transformed input to the system $\transInp(s) = \matEigVect^{-1}
\inp(s)$. Since $\openLoop(s)$ is a scalar transfer function,
(\ref{eq:blockDiagSyst}) is a block diagonal system, where each block has a size
of a Jordan block corresponding to an eigenvalue $\spatEig_i$ of
$\lapl$.
If the Jordan block for the eigenvalue $\spatEig_i$ has a size 1, then it can be
written using a transfer function
\begin{IEEEeqnarray}{rCl}
	\tranFun_i(s)\! =\! \frac{\transPos_i(s)}{\transInp_i(s)}\!=\!
	\frac{\openLoop(s)}{1+\spatEig_i \openLoop(s)}\! =\!
	\frac{\vehNum(s) \contNum(s)}{\vehDen(s) \contDen(s) +\spatEig_i \vehNum(s)
	\contNum(s)}. \label{eq:tranFun}
\end{IEEEeqnarray}
$\tranFun_i(s)$ is an output
feedback system with a feedback gain $\spatEig_i$.
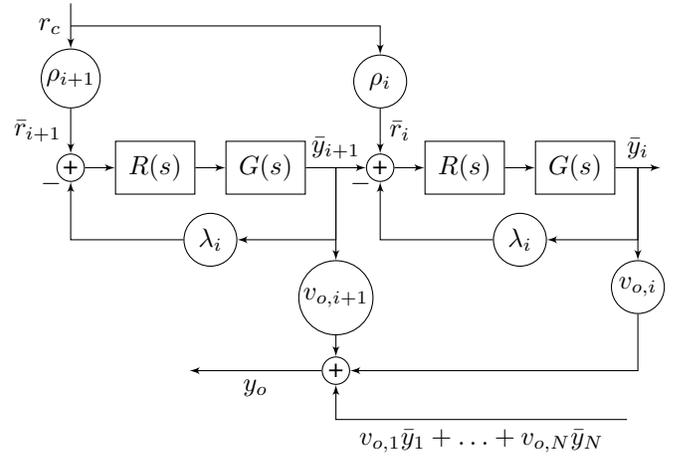
\begin{figure}
\centering
	\begin{tikzpicture}[auto,>=latex',node distance=0.6cm and 0.4cm,]
		\node [input] (input) {};
		\node [gainDown, below=of input] (inpGain1) {$\inputCoef_{i+1}$};
	    \node [sum, below=of inpGain1] (sum) {{+}}; 
	    \node [block, right=of sum] (controller) {$\contTf(s)$};
	    \node [block, right=of controller] (system) {$\vehTf(s)$};
	    \draw [->] (controller) -- node[name=u] {} (system);
	    \node [gainRight, below=of u] (gain) {$\spatEig_i$};
	    
	    \draw [->] (gain) -| node[pos=0.99] {$-$} 
        node [near end] {} (sum); 
        \draw [->] (sum) to (controller);
        \draw [->] (input) to node[name=u, swap] {$\inp_\contNode$}(inpGain1);
        \draw [->] (inpGain1) to node[name=ug1, swap] {$\bar{\inp}_{i+1}$}
        (sum);
    	\node [sum, right=of system, right=0.8cm] (sum2) {{+}};
    	\node [gainDown, above=of sum2] (inpGain2) {$\inputCoef_i$};
    	\node [block, right=of sum2] (controller2) {$\contTf(s)$};
    	\node [block, right=of controller2] (system2) {$\vehTf(s)$};
    	\draw [->] (controller2) to node[name=u2] {} (system2);
    	\node [output, right=of system2, right=0.6cm] (output)  {};
    	\node [gainRight, below=of u2] (gain2) {$\spatEig_i$};
    	\draw [->] (system2) -- node[name=y2] {$\transPos_{i}$} (output); 
    	\draw [->] (y2) |- (gain2); 
    	\draw [->] (gain2) -| node[pos=0.99] {$-$} 
        node [near end] {} (sum2); 
        \draw [->] (sum2) to (controller2);
        \draw [->] (u) -| node[name=ud2]{} (inpGain2);
        \draw [->] (inpGain2) to node[name=ug2] {$\bar{\inp}_{i}$} (sum2) ;
        \draw [->] (system) -- node[name=y1] {$\transPos_{i+1}$} (sum2);
        \draw [->] (y1) |- (gain);
        
        \node [gainDown, below=of y2, below=3.5em] (v2) {$\eigVect_{\obsvNode,
        i}$};
        \node [gainDown, below=of y1, below=3.5em] (v1) {$\eigVect_{\obsvNode,
        i+1}$};
        \node [sum, below=of v1, below=0.8em] (sum3) {{+}};
        \node [output, left=of sum3, left=5em] (output3) {};
        \draw [->] (v1) to (sum3);
        \draw [->] (v2) |- (sum3);
        \draw [->] (y1) to (v1);
        \draw [->] (y2) to (v2);
        \draw [->] (sum3) to node{${\pos}_{\obsvNode}$} (output3);
        \node [placeholder, below=of sum3, below=0.8em] (p1) {};
        \node [input, right=of p1, right=10.5em] (inpOthers) {};
        \draw [->] (inpOthers) -| node[name=others, near start]
        {$\eigVect_{\obsvNode, 1}\transPos_1 + \ldots + \eigVect_{\obsvNode,
        \numVeh}\transPos_\numVeh$} (sum3);

\end{tikzpicture}
	\caption{One diagonal block for the case of Jordan block of size 2. The
	eigenvalue $\spatEig_i$ acts as a gain in the feedback. Only one closed loop is
	present if the Jordan block has a size one. }
	\label{fig:diagonalizedSystemJordan}    
\end{figure}
If, on the other hand, the block in (\ref{eq:blockDiagSyst}) corresponds to
a~Jordan block of size 2, then its output can be written as the output of
a series connection of identical blocks, such as
\begin{IEEEeqnarray}{rCl}
	\transPos_i(s)\! = \! \frac{\openLoop(s)}{1+ \spatEig_i
	\openLoop(s)}\left( \transInp_i(s)\!+\! \frac{\openLoop(s)}{1+ \spatEig_i
	\openLoop(s)} \transInp_{i+1}(s)\right).
\end{IEEEeqnarray}
This easily generalizes to larger Jordan blocks. 
The structure is shown in Fig. \ref{fig:diagonalizedSystemJordan}. 

For simplicity, the derivations throughout the paper will be shown only for the
case where all Jordan block in $\matJ$ are simple --- the eigenvalues
$\spatEig_i$ have the same algebraic and geometric multiplicity.
All the proofs can be conducted the same way for blocks of larger size and all
the results remain valid.

If the eigenvalue $\spatEig_i$ is simple, the input to the $i$th diagonal block
in (\ref{eq:blockDiagSyst}) is the $i$th element of $\bar{\inp}$ and equals
	$\bar{\inp}_{i} = \canonVect_i^T\matEigVect^{-1}\canonVect_{\contNode}
	\inp_\contNode = \inputCoef_i \inp_\contNode$ 
with
$\inputCoef_i=\canonVect_i^T\matEigVect^{-1}\canonVect_{\contNode}=(\matEigVect^{-1})_{i
\contNode}$.
Thus, the input $\inp_\contNode$ enters the block $\diagBlock_i(s)$
through the gain $\inputCoef_i$ and from
(\ref{eq:tranFun}) $\bar{\pos}_i(s)=\tranFun_i(s) \inputCoef_i
\inp_\contNode(s)$.
The output of the $i$th agent can be obtained using the
outputs of the blocks as
	$\pos_i(s) = \sum_{j=1}^{\numVeh} \eigVect_{i j}
	\transPos_j(s).$
By setting $\transPos_j(s)=\tranFun_i(s) \inputCoef_i \inp_\contNode(s)$
in the previous equation, the output of the observing node is
\begin{equation}
	\pos_\obsvNode(s) = \left[\sum_{i=1}^{\numVeh} \eigVect_{\obsvNode i}
	\inputCoef_i \tranFun_i(s)\right] \inp_\contNode(s) = \tranFunCo(s)
	\inp_\contNode(s). \label{eq:outputTranFunObsv}
\end{equation}
This also expresses the transfer function $\tranFunCo(s)$ in
(\ref{eq:tranFunDef}).

\section{Transfer functions in graphs}
In this section we derive the structure of the transfer function
$\tranFunCo(s)$ between the input $\inp_\contNode$ of the controlling node 
and output of the observing node $\pos_\obsvNode$.

\subsection{Single integrator dynamics}
Before investigating the general case with higher-order dynamics, let us discuss
a standard single-integrator case. We will later in the paper relate
it to the higher-order dynamics. For the single single-integrator
case $\openLoop(s) = \frac{1}{s}$ and the state-space description of the network system is
	$\dot{\stateVar} = -\lapl \stateVar + \canonVect_\contNode \inp_\contNode,
	\quad \pos_\obsvNode = \canonVect_\obsvNode^T \stateVar.$
Let the single-integrator transfer function from $\inp_\contNode$ to
$\pos_\obsvNode$ be a fraction of two polynomials as
\begin{equation}
	\tranFunCo(s)=\frac{\numPol(s)}{\denPol(s)}. \label{eq:tranFunSingle}
\end{equation} 
The denominator polynomial $\denPol(s)$ is given as 
\begin{equation}
		\denPol(s)\!=\! \det(s\idMat_{\numVeh} + \lapl)\! = \!
		s^{\numVeh}\! +\! \charPolLaplCoef_{\numVeh-1} s^{\numVeh-1}\! +\! \ldots\! +\!
		\charPolLaplCoef_1 s \!+\! \charPolLaplCoef_0.
		\label{eq:charPoly}
\end{equation}
$\denPol(s)$ is a characteristic polynomial of $-\lapl$. The roots of $\denPol$
(i.
e., the poles of $\tranFunCo(s)$ for single integrator dynamics) are
$-\spatEig_i$, the eigenvalues of $-\lapl$. The coefficient
$\charPolLaplCoef_0=0$ because there is always a zero eigenvalue of $-\lapl$.
If the zero eigenvalue is simple, it is known that the coefficients are
\begin{equation}
	\charPolLaplCoef_{\numVeh-1} = \sum_{i=1}^{\numVeh} \spatEig_i,
	\:\:
	\charPolLaplCoef_{\numVeh-2} = \sum_{\mathclap{i=1, j=1, i \neq j}}^{\numVeh}
	\spatEig_i \spatEig_j, \:\: \ldots, \:\:
	\:\:
	\charPolLaplCoef_1 = \prod_{i=2}^{\numVeh} \spatEig_i.	
	\label{eq:charLaplCoefDef}
\end{equation}
The other
terms $\charPolLaplCoef_k$ are sums of all products of $k$ eigenvalues.

The numerator polynomial is given as $\numPol(s) =
\charPolNumCoefSingle_{\orderQ} s^{\orderQ} + \ldots + \charPolNumCoefSingle_1 s
+ \charPolNumCoefSingle_0$. It was
	shown in \cite{Briegel2011, Torres2013} that $\orderQ=\numVeh - \distanceCO -
	1$. We denote the $\orderQ$ roots of $\numPol(s)$ as
$-\spatEigZ_i$, so
\begin{equation}
	\numPol(s) =
	\charPolNumCoefSingle_{\orderQ}
	(s+\spatEigZ_1)(s+\spatEigZ_2)\ldots(s+\spatEigZ_{\orderQ}).
	\label{eq:numPolGamma} \end{equation} 

The coefficients of $\denPol$ and $\numPol$ have a graph-theoretic
representation. For the denominator polynomial $\denPol(s)$ they are given
by \cite[Proposition 2]{Chebotarev2002} as $
	\charPolLaplCoef_i = \pathWeight\left( \forestSet_{\numVeh-i} \right)$,
which is the weight of the set of all diverging forests in the graph with
$\numVeh -i$ arcs. This also explains why $\charPolLaplCoef_0=0$ --- there is no
spanning forest with $\numVeh$ arcs (there has to be a cycle in $\numVeh$ arcs). 

The numerator polynomial can be calculated
as
\begin{equation}
	\numPol(s) = \canonVect_\obsvNode^T \, \text{adj}(s \idMat + \lapl)\,
	\canonVect_\contNode, \label{eq:numPolAsAdjugateMat}
\end{equation}
which is the $\obsvNode,\contNode$th cofactor of $(s \idMat+\lapl)$. It is shown
in \cite[Proposition 3]{Chebotarev2002} that
\begin{equation}
	\text{adj}(s \idMat + L) = \sum_{i=0}^{\numVeh} Q_i s^{\numVeh-i-1},
	\label{eq:cofactorAsForests}
\end{equation}

\begin{lem}
The coefficients $\charPolNumCoefSingle_i$ are given as
	$\charPolNumCoefSingle_i = \pathWeight(\forestSet_{\numVeh-i-1}^{\contNode
	\rightarrow \obsvNode}).$
\label{lem:coefForestWeights}
\end{lem}
\begin{proof}
	The polynomial $\numPol(s)$ equals the $\obsvNode, \contNode$ element of
	$\text{adj}(s \idMat + \lapl)$ (\ref{eq:numPolAsAdjugateMat}). The
	coefficient at $s^i$ in $\numPol(s)$ is by (\ref{eq:cofactorAsForests}) equal
	to the $\obsvNode, \contNode$ element of matrix $Q_{\numVeh-i-1}$, i.e.,
	$\charPolNumCoefSingle_i = q^{\numVeh-i-1}_{\obsvNode \contNode}$.
	By (\ref{eq:weightOfForestCofactor}) this element also must be equal to
	$\pathWeight(\forestSet_{\numVeh-i-1}^{\contNode \rightarrow \obsvNode})$.
\end{proof}

This indicates that the coefficients $\charPolNumCoefSingle_i$ are given as the
weights of the set of all spanning diverging forests with $\numVeh-i-1$ arcs
which contain $\obsvNode$ and diverge from $\contNode$.
In the case of unweighted graph the weight reduces to the number of such
out-forests.

While the coefficients in the denominator polynomial correspond to all
diverging forests with the given number of arcs, the numerator polynomial takes
only those spanning out-forests containing the controlling and the observing
nodes.

\subsection{Higher order dynamics}
Now let us go back to higher-order systems. We have the definition of
$-\spatEigZ_i$ as the roots of $\charPolNum(s)$ in (\ref{eq:numPolGamma}), so
we can state the main theorem of the paper. It relates the single-integrator
systems to the higher-order dynamics. 
\begin{thm}
The transfer function $\tranFunCo(s)$ can be
written as 
\begin{equation}
	\tranFunCo(s)\!=\!\pathWeightCo
	\frac{[\vehNum(s)\contNum(s)]^{1+\distanceCO}\;
	\displaystyle{\prod_{i=1}^{\mathclap{\numVeh-1-\distanceCO}}}\;
	\big(\vehDen(s)\contDen(s) + \spatEigZ_i \vehNum(s)\contNum(s)\big)}{\displaystyle{\prod_{i=1}^{\numVeh}}
	\Big(\vehDen(s)\contDen(s) + \spatEig_i \vehNum(s)\contNum(s)\Big)},
	\label{eq:tranFunProductForm}
\end{equation}
where $\pathWeight_{\contNode
	\obsvNode}=
	\charPolNumCoefSingle_{\numVeh-\distanceCO-1}$ is the sum of weights of all shortest paths from $\contNode$ to $\obsvNode$, $\distanceCO$ is the
	distance from $\contNode$ to $\obsvNode$ and the gains $-\spatEigZ_i$ defined
	in (\ref{eq:numPolGamma}) is the root of $\numPol(s)$ .
\label{lem:totTranFun}
\end{thm} 
The proof can be found in the appendix. It is clear that the roots
$-\spatEigZ_i$ of the single-integrator numerator polynomial $\numPol(s)$ have
the same role as the roots $-\spatEig_i$ of the denominator polynomial
$\denPol(s)$. As can be seen, the structure of the terms in the numerator and
the denominator of (\ref{eq:tranFunProductForm}) is $\vehDen(s) \contDen(s) + k
\,\,\vehNum(s) \contNum(s)$, where $k=\spatEig_i$ in the denominator and
$k=\spatEigZ_i$ in the numerator. In addition, such structure is the same as the
structure of the characteristic polynomial of an output-feedback system with the
open loop $\openLoop(s)=k \, \frac{\vehNum(s) \contNum(s)}{\vehDen(s) \contDen(s)}$ with the gain
$k=\spatEig_i$ or $k=\spatEigZ_i$.

If both $\spatEigZ_i$ and $\spatEig_i$ are real, the poles and
zeros of (\ref{eq:tranFunProductForm}) lie on the root-locus curve (see Fig.
\ref{fig:exampleGraphPoles} for an example). The root-locus curve is
defined as a location of roots of $\vehDen(s) \contDen(s) + k \, \vehNum(s)
\contNum(s)$ as a function of $k \in (0, \infty)$.
Note that both the terms in the numerator and denominator of
(\ref{eq:tranFunProductForm}) have this form.

A particular case of the product form (\ref{eq:tranFunProductForm}) was shown in
\cite[Proposition 3]{Lin2012a}, where the authors considered single integrators ($\openLoop(s)=1/s$) and unidirectional interaction.

The product form in  (\ref{eq:tranFunProductForm}) can be written also as
\begin{equation}
	\tranFunCo(s)=\pathWeightCo
	\prod_{i=1}^{\numVeh-\distanceCO-1}\tranFunZ_i(s)
	\prod_{j=\numVeh-\distanceCO}^{\numVeh}\tranFun_j(s),	
	\label{eq:productZiTi}
\end{equation}
with  
\begin{IEEEeqnarray}{rCl}
	\tranFunZ_i(s) &=&
	\frac{\vehDen(s)\contDen(s)+\spatEigZ_i\vehNum(s)\contNum(s)}{\vehDen(s)\contDen(s)+\spatEig_i
	\vehNum(s)\contNum(s)}, \\ 
	\tranFun_j(s) &=& \frac{\vehNum(s)\contNum(s)}{\vehDen(s)\contDen(s)+\spatEig_j
	\vehNum(s)\contNum(s)}.
\end{IEEEeqnarray}
All eigenvalues $\spatEig_i$ must be used, so $\numVeh-\distanceCO-1$ of them go
to $\tranFunZ_i(s)$ and the remaining $\distanceCO+1$ to $\tranFun_j(s)$. The
transfer functions $\tranFunZ_i(s)$ are biproper and the numerator differs from the
denominator only in the multiplication factor $\spatEigZ_i$. The transfer
functions $\diagBlock_j(s)$ are standard output feedback systems in
(\ref{eq:tranFun}).

The network system (\ref{eq:overallSyst}) of identical agents with arbitrary
interconnection was transformed in equation (\ref{eq:productZiTi}) to a series
connection (product of transfer functions) of non-identical (but structured)
subsystems. In many cases, such as in determining a frequency response, the
series connection is much easier to analyze \cite{Herman2013b}. The series connection is illustrated in Fig.
\ref{fig:seriesForm}.

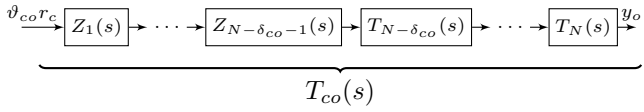
\begin{figure}[t]
\centering
	\begin{tikzpicture}[auto, >=latex']
\scriptsize
	\node[placeholder] (inp) {};
	\node[block]  (t1) [right=of inp, right=2.3em] {$\tranFunZ_1(s)$};	
	\draw [->] (inp) to node[name=u, near start] {$\pathWeightCo \inp_\contNode$}(t1); 
	\node  (dots1)	[right=of t1, right=1em] {$\ldots$}; 
	\draw [->] (t1) to (dots1); 
	\node[block]  (t2) [right=of dots1, right=1em] 	{$\tranFunZ_{\numVeh-\distanceCO-1}(s)$}; 
	\draw [->] (dots1) to (t2); 
	\node[block]  (z1) [right=of t2,	right=1em]	{$\tranFun_{\numVeh-\distanceCO}(s)$}; 
	\draw [->] (t2) to (z1);
	\node (dots2) [right=of z1,	right=1em] {$\ldots$}; 
	\draw [->] (z1) to (dots2);  
	\node[block]  (z2) [right=of dots2,	right=1em]	{$\tranFun_{\numVeh}(s)$}; 
	\draw [->] (dots2) to (z2); 
	\node[placeholder] (out) [right=of z2, right=1em] {};
	\draw [->] (z2) to node[name=y, near end] {$\pos_\obsvNode$}(out);
	\draw[decorate,decoration={brace},thick] (8.4,-0.5) to
	node[midway,below=0.4em] (bracket) {\normalsize $\tranFunCo(s)$} (0.4,-0.5);
\end{tikzpicture}
	\caption{Series form of the transfer function $\tranFunCo(s)$.}
	\label{fig:seriesForm}
\end{figure}

As the numerator of the open loop $\vehNum(s)\contNum(s)$ is present
for $\distanceCO+1$ times in (\ref{eq:tranFunProductForm}), we have the
following corollary.
\begin{cor}
	The transfer function $\tranFunCo(s)$ has $\distanceCO+1$ multiple zeros
	at the locations of the zeros of the open loop, i. e. roots of $\vehNum(s)
	\contNum(s)=0$.
	\label{cor:openLoopZeros} 
\end{cor}
These zeros can be partly chosen by the designer of the network, since he can
choose the controller numerator $\contNum(s)$ freely. Contrary, the
zeros of $\tranFunZ_i(s)$ are given by the interconnection matrix in the same way as
the poles are. 

A relative degree comes immediately from Theorem \ref{lem:totTranFun}.
\begin{cor}
	Let $\relOrd$ be the relative degree of $\openLoop(s)$. Then the relative
	degree $\relOrd_{\contNode \obsvNode}$ of $\tranFunCo(s)$ is
		$\relOrd_{\contNode \obsvNode} =  (\distanceCO + 1)\relOrd.$
	\label{cor:relDegree}
\end{cor}
\begin{proof}
	There is $\numVeh-\distanceCO-1$ blocks of type $\tranFunZ_i(s)$ in
	(\ref{eq:tranFunProductForm}), which have relative degree 0. Then there
	is $\distanceCO+1$ terms $\tranFun_i(s)$ which have relative degree $\relOrd$.
	Hence, $\relOrd_{\contNode \obsvNode}=(\distanceCO+1)\relOrd$.
\end{proof}
The relative degree strongly affects the transients. The transfer functions
$\tranFunZ_{i}(s)$ have relative order 0, so the input gets directly to the
output. The $\distanceCO+1$ terms $\tranFun_j(s)$ slow down the
transient. Quite clearly, the further the control and observer nodes are from
each other, the slower the transient will be.

Another immediate result is the steady-state value.
\begin{cor}
	For at least one integrator in the open loop, the steady-state gain of any
	transfer function in the formation is
	\begin{equation}
		\tranFunCo(0) = \pathWeightCo \frac{\prod_{i=1}^{\numVeh-1-\distanceCO}
		\spatEigZ_i}{\prod_{i=1}^{\numVeh} \spatEig_i}.
	\end{equation}
	\label{cor:dcgain}
\end{cor}
\begin{proof}
	For at least one integrator in the open loop, $\vehDen(0)\contDen(0)=0$. After
	plugging this to (\ref{eq:tranFunProductForm}), the result follows.
\end{proof}
At least one integrator in the open loop is a common requirement to allow an
uncontrolled network system to have a nonzero equilibrium. 

 The most important fact
following from the Corollary \ref{cor:dcgain} is that the steady-state gain \emph{does not depend} on the open-loop model, as long as
there is at least one integrator in $\openLoop(s)$. To change the steady-state value,
the interconnection structure must be modified.

We will discuss two cases. First, assume that $\spatEigZ_i \neq 0, \:
\forall i$. 
Then the eigenvalue $\spatEig_1=0$ of the Laplacian in the denominator makes the
steady-state gain infinite. This happens when there is no independent leader in
the network system.

If, on the other hand, there is $\spatEigZ_1=0$, the eigenvalue at the origin
$\spatEig_1=0$ will be cancelled. As a result, the steady-state value is bounded. 
The presence of $\spatEigZ_1=0$ is usually caused by the presence of an
independent leader in the system. Such a leader cannot
be controlled from the network system, hence the zero eigenvalue will be
uncontrollable, causing the pole-zero cancellation.

\section{General transfer functions}
\label{sec:otherInputs}
So far we have analyzed properties of a transfer function from the input of the
controller of agent $\contNode$ to the output of the agent $\obsvNode$. However,
we might also be interested in a transfer function from a general input $\inpGen_\contNode$ at the
controlling node to a general output $\outGen_\obsvNode$ of the observing node.
In this section we show that the general transfer function has two parts: an
open-loop part and a network part.

There is always at least one zero
eigenvalue of $\lapl$, therefore in (\ref{eq:tranFunProductForm})
$\vehDen(s)\contDen(s)+\spatEig_1 \vehNum(s)\contNum(s)=\vehDen(s)\contDen(s),$
 which is the denominator of the open loop $\openLoop(s)$.
Also at least one numerator polynomial of the open loop $\vehNum(s)\contNum(s)$ is present in
$\tranFunCo(s)$.
Then the transfer function in (\ref{eq:tranFunProductForm}) can be written
as
\begin{IEEEeqnarray}{rCl}
	\tranFunCo(s) &=& \pathWeightCo \openLoop(s)
	\frac{\big(\vehNum(s)\contNum(s)\big)^{\distanceCO}\displaystyle
	\prod_{j=1}^{\orderQ} \vehDen(s)\contDen(s) + \spatEigZ_j \vehNum(s) \contNum(s)}{\prod_{i=2}^{\numVeh} \vehDen(s)\contDen(s) + \spatEig_i
	\vehNum(s) \contNum(s)} \nonumber \\ 
	&=& \openLoop(s) \tranFunFeedbackPartCo(s),
	\label{eq:tranFunFeedback}
\end{IEEEeqnarray}
where 
\begin{equation}
 \tranFunFeedbackPartCo(s) = 
\pathWeightCo \frac{\big(\vehNum(s)\contNum(s)\big)^{\distanceCO}
\prod_{j=1}^{\orderQ} \big(\vehDen(s)\contDen(s) + \spatEigZ_j \vehNum(s)
\contNum(s)\big)}{\prod_{i=2}^{\numVeh} \big(\vehDen(s)\contDen(s) + \spatEig_i
\vehNum(s) \contNum(s)\big)}
\label{eq:networkPart}
\end{equation}
is the network part of
$\tranFunCo(s)$ and $\openLoop(s)$ is the open-loop.

Let $\openLoopPart(s)$ be the transfer function in open-loop of one agent from
the desired input $\inpGen_i$ (e. g. a reference or a disturbance) to
the desired output $\outGen_i$ of same agent, i. e.,
$\openLoopPart(s)=\outGen_i(s) / \inpGen_i(s)$.
\begin{thm}
	The transfer function $\tranFunGenInpState{\contNode}{\obsvNode}(s)$ from the
	input of the controlling agent $\inpGen_\contNode$ to the output
	$\outGen_\obsvNode$ of the observing agent is given as
	\begin{equation}
		\tranFunGenInpState{\contNode}{\obsvNode}(s) =
		\frac{\outGen_\obsvNode(s)}{\inpGen_{\contNode}(s)}=
		\openLoopPart(s)\tranFunFeedbackPartCo(s).
	\end{equation}
	\label{thm:genInputState}
\end{thm}
\begin{proof}
	Consider first that the controlling and the observing nodes are collocated
	($\contNode=\obsvNode$).
	Then by changing the input from $\inp_\contNode$ to $\inpGen_\contNode$ and
	the output from $\pos_\obsvNode$ to $\outGen_\obsvNode$ we just change the
	direct branch of the transfer function
	$\tranFunFromTo{\contNode}{\contNode}(s)$.
	The direct branch is then $\openLoopPart(s)$ instead of $\openLoop(s)$. The
	network (feedback) part $\tranFunFeedbackPart_{\contNode \contNode}$ in
	(\ref{eq:networkPart}) remains unchanged.
	That is,
	\begin{equation}
		\frac{\outGen_\contNode(s)}{\inpGen_\contNode(s)}=\openLoopPart(s)\tranFunFeedbackPart_{\contNode
		\contNode}(s). \label{eq:tranFunGenCollocated}
	\end{equation}
	
	Consider now that $\contNode$ and $\obsvNode$ are not collocated. Define two
	transfer functions of a single agent:
	\begin{equation}
		\openLoop_1(s)=\frac{\pos_i(s)}{\inpGen_i(s)}, \quad \openLoop_2(s) =
		\frac{\outGen_i(s)}{\inp_i(s)}. 
	\end{equation} 
	Note that $\frac{\openLoop_1(s)\openLoop_2(s)}{\openLoop(s)} =
	\frac{[\pos_i(s)/\inpGen_i(s)]\,[\outGen_i(s)/\inp_i(s)]}{\pos_i(s)/\inp_i(s)}=\openLoopPart(s)$.
	
	The transfer function from $\pos_\contNode(s)$ to
	$\pos_\obsvNode(s)$ using the input $\inp_\contNode$ is
	\begin{equation}
		\frac{\pos_\obsvNode(s)}{\pos_\contNode(s)} = \frac{\inp_\contNode(s)
		\openLoop(s) \tranFunFeedbackPartCo(s)}{\inp_\contNode(s)
		\openLoop(s) \tranFunFeedbackPart_{\contNode \contNode}(s)} =
		\frac{\tranFunFeedbackPartCo(s)}{\tranFunFeedbackPart_{\contNode
		\contNode}(s)}. \label{eq:tranFunOutputOutput}
	\end{equation}
	From (\ref{eq:tranFunGenCollocated}) we get
	$\pos_\contNode(s)=\openLoop_1(s)\tranFunFeedbackPart_{\contNode \contNode}
	\inpGen_\contNode(s)$. Plugging this to (\ref{eq:tranFunOutputOutput}) gives
	\begin{equation}
		\pos_\obsvNode(s) =  \frac{\tranFunFeedbackPartCo(s)}{\tranFunFeedbackPart_{\contNode
		\contNode}(s)} \openLoop_1(s)\tranFunFeedbackPart_{\contNode \contNode}(s)
	\inpGen_\contNode(s) = \openLoop_1(s) \tranFunFeedbackPartCo(s)
	\inpGen_\contNode(s).
	\label{eq:tranFunGenInpToOutput}
	\end{equation}
	
	Similarly, the transfer function from $\outGen_\obsvNode(s)$
	to $\pos_\obsvNode(s)$ is
	\begin{equation}
		\frac{\pos_\obsvNode(s)}{\outGen_\obsvNode(s)} = \frac{\inp_\obsvNode(s)
		\openLoop(s) \tranFunFeedbackPart_{\obsvNode \obsvNode}(s)}{\inp_\obsvNode(s)
		\openLoop_2(s) \tranFunFeedbackPart_{\obsvNode \obsvNode}(s)} =
		\frac{\openLoop(s)}{\openLoop_2(s)}, \label{eq:tranFunOutputGenOutput}
	\end{equation}
	therefore $\pos_\obsvNode(s) = \openLoop(s)/\openLoop_2(s)
	\outGen_{\obsvNode}(s)$. Plugging this to (\ref{eq:tranFunGenInpToOutput})
	and separating $\outGen_\obsvNode(s)$ yields 
	\begin{equation}
		\outGen_\obsvNode(s) = \frac{\openLoop_1(s)
		\openLoop_2(s)}{\openLoop(s)}\tranFunFeedbackPartCo(s)\inpGen_\contNode(s) =
		\openLoopPart(s) \tranFunFeedbackPartCo(s)\inpGen_\contNode(s).
	\end{equation}
	The transfer function  $\tranFunGenInpState{\contNode}{\obsvNode}(s)$ follows.
\end{proof}

The general structure is shown in Fig. \ref{fig:twoPartsOfTF}. It follows that
each transfer function in the network system is given by two parts:
\begin{enumerate}
  \item the network part $\tranFunFeedbackPartCo(s)$, which is the same
for all transfer functions with the same $\contNode$ and $\obsvNode$ nodes
and is given by the interconnection,
	\item the open loop part $\openLoopPart(s)$, which depends on the inputs and
	outputs of interest.
\end{enumerate}

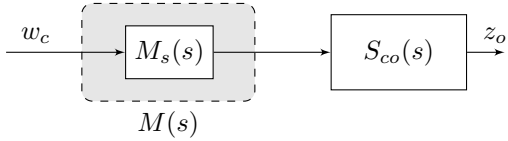
\begin{figure}[t]
\centering
	\begin{tikzpicture}[auto, >=latex']
	\node[placeholder] (inp) {};
	\node[block,minimum width=2.3cm, minimum
	height=1.3cm, fill=black!10,draw,dashed,rounded corners] (ol) [right=of inp,
	label=below:{$\openLoop(s)$}] {}; \node[block]  (olred) [right=of inp,
	right=4.5em] {$\openLoopPart(s)$}; \draw [->] (inp) to node[name=u, near start]
	{$\inpGen_\contNode$}(olred); \node[block,minimum width=1.8cm, minimum
	height=1.0cm,] (fb) [right=of ol] {$\tranFunFeedbackPartCo(s)$}
		edge [<-] (olred);
	\node[placeholder] (out) [right=of fb] {};
	\draw [->] (fb) to node[name=y, near end] {$\outGen_\obsvNode$}(out);
\end{tikzpicture}
	\caption{Two parts of transfer functions between $\contNode$ and $\obsvNode$
	for general input and output.}
	\label{fig:twoPartsOfTF}
\end{figure}

\subsection{Disturbances}
First we analyze an input disturbance $\inDist_\contNode$, acting at the input
of the plant.
The modified open-loop transfer function is
$\openLoopPart(s)=\vehTf(s)$. Then the transfer function is
\begin{equation}
	\tranFunDistIn{\contNode}{\obsvNode}(s)=\frac{\pos_\obsvNode(s)}{\inDist_\contNode(s)}
	= \vehTf(s) \tranFunFeedbackPartCo(s). 
\end{equation} 

It is clear that $\tranFunCo(s)$  and
$\tranFunDistIn{\contNode}{\obsvNode}(s)$ differ only in the presence of
transfer function of the controller and $\tranFunCo(s) = \contTf(s) \tranFunDistIn{\contNode}{\obsvNode}(s)$. 

The output disturbance $\outDist$ changes the output of the plant of the $j$th
agent as
	$\pos_j = \bar{\pos}_j+\outDist_j,$
where $\bar{\pos}_i$ is the output of the agent without disturbance. In
this case $\openLoopPart(s)=1$, so the transfer function for output disturbance
is
\begin{equation}
	\tranFunDistOutCo(s) = \frac{\pos_\obsvNode(s)}{\outDist_\contNode(s)} =
	\tranFunFeedbackPartCo(s).
\end{equation}

 
\section{Relations to single-integrator case}
In this section we provide some results for the single-integrator case. They
easily generalize to higher-order dynamics, because of the fact that
$\spatEigZ_i$, the gain in the closed loop in (\ref{eq:tranFunProductForm}), is
the same as the zero in the single-integrator dynamics. Let us denote
$\redLapl_{i:j}^k$ as a matrix which is obtained from $\lapl$ by deleting the
rows and columns corresponding to the vertices on the $k$th path from vertex $i$
to $j$.

 The simplest case is when the controller node and
 observer nodes are collocated, i.e.
$\contNode = \obsvNode$. Then, as shown in \cite{Briegel2011, Torres2014,
Herman2014a}, the zeros are given as eigenvalues of $\redLapl^1_{(\contNode:\contNode)}$ and the
numerator polynomial is
\begin{equation}
	 \numPol(s) = \det (s \idMat +
	 \redLapl^1_{(\contNode:\contNode)}). \label{eq:contObsvcollocated} 
\end{equation}
The spectrum of this reduced Laplacian (also known as
a grounded Laplacian) is discussed in \cite{Pirani2014}.

The next theorem was independently discovered in \cite{Torres2013}
using purely algebraic techniques. Here we provide a graph-theoretic proof.
\begin{thm}
	If there is only one path between the controlling node and the observing node,
	then 
	\begin{IEEEeqnarray}{rCl}
		\charPolNum(s) = \pathWeightCo \det
		\left(s \idMat + \redLapl^1_{\contNode:\obsvNode}
		\right).
	\end{IEEEeqnarray}	
	The roots $-\spatEigZ_i$ of $\charPolNum(s)$ are
	the eigenvalues of $-\redLapl^1_{\contNode:\obsvNode}$.
	\label{conj:onePath}
\end{thm}

\begin{proof}
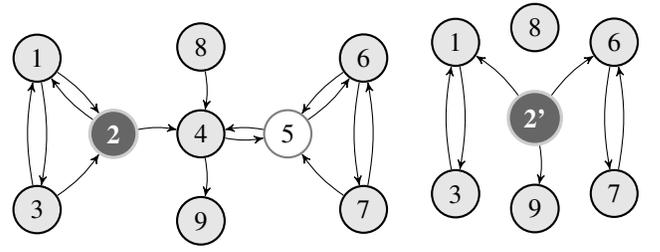
\begin{figure}[t]
\centering
\begin{subfigure}{0.30\textwidth}
	\begin{tikzpicture}
	\node[normalNode]  (n1) {1};
	\node[placeholder] (p1) [below=of n1] {};
	\node[controllingNode] (n2) [right=of p1]{2}
		edge [<-, bend right] (n1)
		edge [->, bend left] (n1);
	\node[normalNode] (n3) [below=of p1]{3}
		edge [->, bend left] (n1)
		edge [<-, bend right] (n1)
		edge [->, bend right] (n2);
	\node[normalNode] (n4) [right=of n2]{4}
		edge [<-, bend right] (n2);
	\node[observingNode] (n5) [right=of n4] {5}
		edge [->, bend right] (n4)
		edge [<-, bend left] (n4);
	\node[placeholder] (p2) [right=of n5] {};
	\node[normalNode] (n6) [above=of p2] {6}
		edge [->, bend right] (n5)
		edge [<-, bend left] (n5);  
	\node[normalNode] (n7) [below=of p2] {7}
		edge [->, bend right] (n6)
		edge [<-, bend left] (n6)
		edge [->, bend left] (n5);
	\node[normalNode] (n8) [above=of n4] {8}
		edge [->, bend left] (n4);
	\node[normalNode] (n9) [below=of n4] {9}
		edge [<-, bend right] (n4);
\end{tikzpicture}
	\caption{Before reduction}
	\label{fig:beforeRed}
\end{subfigure}
\begin{subfigure}{0.15\textwidth}
	\begin{tikzpicture}
	\node[normalNode]  (n1) {1};
	\node[placeholder] (p1) [below=of n1] {};
	\node[controllingNode] (n2) [right=of p1]{2'}
		edge [->, bend right] (n1);
	\node[normalNode] (n3) [below=of p1]{3}
		edge [->, bend left] (n1)
		edge [<-, bend right] (n1);
	\node[placeholder] (p2) [right=of n2] {};
	\node[normalNode] (n6) [above=of p2] {6}
		edge [<-, bend right] (n2); 
	\node[normalNode] (n7) [below=of p2] {7}
		edge [->, bend right] (n6)
		edge [<-, bend left] (n6);
	\node[normalNode] (n8) [above=of n2] {8};
	\node[normalNode] (n9) [below=of n2] {9}
		edge [<-, bend right] (n2);
\end{tikzpicture}
	\caption{After reduction}
	\label{fig:afterRed}
\end{subfigure}
\caption{Graph reduction without changing $\charPolNum(s)$. The controlling node
is the node $2$, observing is $5$.}
	\label{fig:graphReduction} 
\end{figure}
Recall that by (\ref{eq:numPolAsAdjugateMat})
$\charPolNum(s)$ equals $(\obsvNode, \contNode)$ cofactor of $s
\idMat+\lapl$.
By Lemma \ref{lem:coefForestWeights} coefficients $\charPolNumCoefSingle_i$
of $\numPol(s)$ are the weights of the set of all spanning diverging forests
with the root $\contNode$ and containing $\obsvNode$ having $\numVeh-i-1$ arcs,
therefore $\charPolNumCoefSingle_i = 0$ for $i \geq \numVeh-\distanceCO$. In addition, the path from $\contNode$ to $\obsvNode$ must be present in
every spanning forest with more than $\distanceCO$ arcs. 

The proof will be shown in several steps of modifying the original graph
$\graph$ and constructing a new one $\graph^{'}$ with the preserved polynomial
$\charPolNum(s)$.

\begin{enumerate}
  \item Remove all the arcs converging to the path
  $\path{\contNode}{\obsvNode}$ from $\contNode$ to $\obsvNode$. They cannot be
  part of any forest diverging from $\contNode$ and containing $\obsvNode$.
  \item Since by the assumption there is only one path between $\contNode$ and
  $\obsvNode$, the path $\path{\contNode}{\obsvNode}$ is
  present in each forest in Lemma \ref{lem:coefForestWeights} and the weight
  $\pathWeightCo=\pathWeight(\path{\contNode}{\obsvNode})$ of the path must be
  present in all coefficients $\charPolNumCoefSingle_i$. We can write
  \begin{IEEEeqnarray}{rCl} 
	\charPolNum(s) &=& \pathWeightCo
	\bar{\charPolNum}(s) = \pathWeightCo\big(s^{\numVeh-1-\distanceCO}  
	\label{eq:factoredNumCoef}
	\\
	&+&
	\charPolNumCoefMod_{\numVeh-2-\distanceCO} s^{\numVeh-2-\distanceCO} +
	\ldots+ \charPolNumCoefMod_{0}\big). \nonumber
\end{IEEEeqnarray}
This factoring acts as removing the arcs on the path from the graph.
	\item Now we want to find a matrix of which $\bar{\charPolNum}(s)$ is a
	characteristic polynomial. By factoring the weight of the path, we identified (created one from
many) the vertices on the path into only one new vertex $\contNode '$. All arcs
connected to the path are now connected to the new vertex $\contNode '$.
The controlling and observing nodes were collocated. Denote such a new graph as
$\graph^{'}$ with the number of vertices $\numVeh(\graph ')=\numVeh-\distanceCO$. The
process of such graph reduction is illustrated in Fig.
\ref{fig:graphReduction}.
\item The coefficients $\charPolNumCoefMod_i$ in
(\ref{eq:factoredNumCoef}) are the weights of the set of all spanning forests in
the reduced graph $\graph^{'}$, diverging from $\contNode '$ with
$\numVeh(\graph ')-i-1$ arcs.
Then, by (\ref{eq:numPolAsAdjugateMat}-\ref{eq:cofactorAsForests}), the
polynomial $\bar{\charPolNum}(s)$ equals the $(\contNode ',\contNode ')$ cofactor of
$(s\idMat_{\numVeh-\distanceCO} + \redLapl^1_{\contNode:\obsvNode-1})$.
Since the observing and controlling nodes are collocated in the modified
graph $\graph '$, we can use (\ref{eq:contObsvcollocated}) to remove also the
node $\contNode '$ from the graph. 
\end{enumerate} 

In step 3 we deleted all nodes on the path except for the node
$\contNode$. In the last step we were also able to eliminate the
controlling node, so the polynomial $\charPolNum(s)$ can be calculated as 
\begin{equation}
	\charPolNum(s) = \pathWeightCo
	\det (s \idMat + \redLapl^1_{\contNode:\obsvNode}). \qedhere 
	\label{eq:removedVertices}
\end{equation}
\end{proof}
The theorem allows to find $\spatEigZ_i$ directly from the
submatrix of the Laplacian.
The real part of $\spatEigZ_i$ is positive, since the matrix
$\redLapl^1_{\contNode: \obsvNode}$ is still an M-Matrix \cite{Horn1999}. 	In
addition, if $\lapl$ is a symmetric matrix and the conditions in Theorem
\ref{conj:onePath} hold, then $\spatEigZ_i$ interlace with $\spatEig_i$ due to
the Cauchy interlacing
	theorem \cite{Horn1996}.
	
The second theorem is an extension of the previous one.
\begin{thm}
	Let $p(\graph)_{\contNode, \obsvNode}$ be the number of paths from
	the node $\contNode$ to the node $\obsvNode$. Then the numerator polynomial
	$\charPolNum(s)$ in (\ref{eq:tranFunSingle}) is given as a sum of
	characteristic polynomials of $\redLapl_{\contNode:\obsvNode}^i$ corresponding
	to the individual paths $\path{\contNode}{\obsvNode}^i$, i.
	e.
	\begin{equation}
		\charPolNum(s) = \sum_{i=1}^{p(\graph)_{\contNode, \obsvNode}}
		\pathWeight(\path{\contNode}{\obsvNode}^i) \det (s \idMat +
		\redLapl^i_{\contNode:\obsvNode}),
		\label{eq:charPolMultPaths}
	\end{equation}
	\label{conj:sumPath}
\end{thm}
\begin{proof}
Since there are $p(\graph)_{\contNode, \obsvNode}$ paths between the nodes,
there are also $p(\graph)_{\contNode, \obsvNode}$ basic trees diverging from
$\contNode$ and containing $\obsvNode$ (they can have different lengths). For
each of the paths Theorem \ref{conj:onePath} must hold. Let us denote the weight
of spanning forests with $\numVeh-\distanceCO^k-1-i$ arcs corresponding to the
path $k$ with length $\distanceCO^k$ as $\charPolNumCoefSingle_i^k$.
Since the paths are distinct, also the spanning forests corresponding to the
paths will be distinct and the total weight of the set $\forestSet_{k}^{i
\rightarrow j}$ is the sum of the weights of the individual trees. Then each
coefficient in $\charPolNum(s)$ is a sum of the weights of the trees
corresponding to each path, i. e.
	\begin{equation}
		\charPolNumCoefSingle_i =  \sum_{k=1}^{p(\graph)_{\contNode,
	\obsvNode}} \charPolNumCoefSingle_i^k. \label{eq:sumOfCoefsMultPaths}
	\end{equation} 
	Equation (\ref{eq:charPolMultPaths}) then follows from
	(\ref{eq:sumOfCoefsMultPaths}) using Theorem \ref{conj:onePath}.
\end{proof} 


\subsection{Multiple controlling nodes}
Instead of one controlling node $\contNode$ we can have a set
$\contNodeSet=\{\contNode_1, \contNode_2, \ldots, \contNode_{\contNodeNum}\}$ of
$\contNodeNum$ controlling nodes to which the same signal is fed (for instance,
the leader connected to more agents).
Then the numerator polynomial is simply given as a sum of polynomials for individual controlling nodes.
\begin{lem}
	The polynomial $\charPolNum(s)$ for the set of controlling nodes $\contNodeSet$ is equal to
		$\charPolNum(s) = \sum_{i=1}^{\contNodeNum} \charPolNum_i(s)$,
	where $\charPolNum_i(s)$ is the polynomial when the input is fed only to the
	$i$th agent.
	\label{lem:multipleContNode}
\end{lem}
\begin{proof}
	The proof can be obtained using the same arguments of mutually exclusive
	forests as in the proof for Theorem \ref{conj:sumPath}.
\end{proof}

Suppose that $\contNode_n \in \contNodeSet$ is the node 
	in $\contNodeSet$ with the shortest distance to the observing node. Then the 
	relative degree of the transfer function $\tranFunCo(s)$  between
	$\contNodeSet$ and $\obsvNode$ with agents having higher order-dynamics is
	$	\relOrd_{\contNode\obsvNode}  =  (\distance_{\contNode_n \obsvNode} +
		1)\relOrd.
	$ 
This follows since the degree of the sum of polynomials is the
degree of the polynomial of the highest degree.


\subsection{Minimal dimension of a controllable subspace}
\begin{table}
\centering 
\caption{Controllable subspaces for some typical
	undirected graphs with $\numVeh$ vertices.}	
	\label{tab:minContSubs}
\captionsetup{width=0.45\textwidth}
	\begin{tabular}{|c|c|c|c|}
	\hline
	Graph & $\contNode$ node & $\max_{i} d_{\contNode i}$ & Dim. of ctrb. subs. \\
	\hline 
	Star graph & central & 1 & 2 \\
	\hline
	Path graph & end node & $\numVeh-1$ & $\numVeh$ \\
	\hline
	Path graph & central node & $\numVeh/2$ & $\numVeh/2+1$ \\
	\hline
	\end{tabular}	
\end{table}
From equation (\ref{eq:tranFunProductForm}) it follows that if the
single-integrator case is uncontrollable, so are all the systems with higher
order dynamics (we use an output feedback). The following result is an extension
of \cite[Thm.
2]{Zhang2014} to directed graphs.
\begin{thm}
Let $\max_{i} d_{\contNode i}$ be the maximal distance to some of the
other nodes from the controlling node $\contNode$. Then for the dimension 
of the controllable subspace $\mathrm{rank}(\contMat)$ of single integrator
dynamics holds
$\mathrm{rank}(\contMat)\geq \max_{i} d_{\contNode i}+1$.
\label{thm:minContSubsp}
\end{thm}
\begin{proof}
Let us denote the furthest node from $\contNode$ as $f$ and the distance of
$f$ from $\contNode$ as $d_f = \max_{i} d_{\contNode i}$. Let $\lapl^i_\contNode$ denote the $\contNode$th column in $\lapl^i$. Let the vertices on the shortest path from
$\contNode$ to $f$ be labeled as $\vertex_0, \ldots, \vertex_{d_f}$ and the
distance of $\vertex_i$ from $\contNode$ as $\distance_i$.
By Lemma \ref{lem:numwalks} the $\vertex_i$th element in $\lapl^j_\contNode$ is
zero for all $j<d_i$ and is nonzero for $j \geq d_i$. Therefore,
$\lapl^{d_i}_\contNode$ is linearly independent of $\lapl^j_\contNode$ for $j <
d_i$ and $d_i=0,\ldots, d_f$.
Consequently, all columns $[\lapl^0_\contNode, \lapl^1_\contNode, \ldots,
\lapl^{d_f}_\contNode]$
 must be linearly independent.

The controllability criterion matrix is defined as
$	\contMat = [\lapl^0_\contNode, \lapl^1_\contNode, \ldots,
	\lapl^{\numVeh}_\contNode].$ 
By previous
development we know that at least $\lapl^0_\contNode, \ldots
\lapl^{d_f}_\contNode$ are linearly independent, hence $\text{rank}(\mathcal{C})
\geq d_f+1$.
\end{proof}
Of course, the controllable subspace can be much greater than indicated by
this theorem and our result can be very conservative. The bound is achieved for
some graphs and controlling nodes, as shown in Table
\ref{tab:minContSubs}. Some further discussion of the tightness of the bound is
in \cite[Remark 2]{Zhang2014}.
Theorem \ref{thm:minContSubsp} gives a strong structural controllability, since
it does not depend on the weights of the arcs.
By any choice of the nonzero weights of arcs, the controllable subspace cannot
have smaller dimension than $\max_{i} d_{\contNode i}+1$. Structural controllability
is described, e.g., in \cite{Chapman2013, Clark2014}.

Surprisingly, the more distant node exists in a graph, the greater the guaranteed
dimension of the controllable subspace. On the other hand, it was shown in
\cite{Shi2014} that the transient time grows with the maximal
distance from the control node. Similarly, at least for a path graph it follows
from \cite{Fitch2013} that the external input should be applied to the agent where it minimizes the
maximal distance. This is also confirmed by the relative degree in Corollary
\ref{cor:relDegree} --- the higher the degree, the slower is the information
spread. However, in this case the node has the smallest guaranteed controllable
subspace. An optimization procedure for the tradeoff between performance and
controllability is presented in \cite{Clark2014}.

\section{Illustrative example}
Consider a directed and weighted graph with five nodes shown in Fig.
\ref{fig:graphPath} (The arcs without a weight shown
	have a weight one).
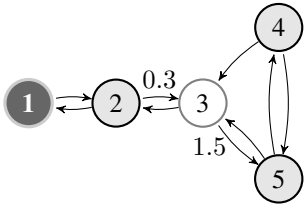
\begin{figure}  
\centering 
	\begin{tikzpicture}[->,>=stealth',shorten >=1pt,auto,node distance=3cm]
	\node[controllingNode]  (n1) {1};
	\node[normalNode] (n2) [right=of n1]{2}
		edge [<-, bend right] (n1)
		edge [->, bend left] (n1);
	\node[observingNode] (n3) [right=of n2]{3}
		edge [->, bend left] (n2)
		edge [<-, bend right] node[above] {$0.3$} (n2);
	\node[placeholder] (p2) [right=of n3] {};
	\node[normalNode] (n4) [above=of p2] {4}
		edge [->, bend right] (n3);
	\node[normalNode] (n5) [below=of p2] {5}
		edge [->, bend left]  (n4)
		edge [<-, bend right]  (n4)
		edge [->, bend right] (n3)
		edge [<-, bend left]  node[left] {$1.5$} (n3);
\end{tikzpicture}
	\caption{Directed graph used in the example. }
	\label{fig:graphPath}    
\end{figure}
The plant is $\vehTf(s)=1/s$, the controller is $\contTf(s)=(s+1)/s$ (a PI
controller applied to a single integrator). 
The open-loop model is
	$\openLoop(s) = \frac{s+1}{s^2}.$
Let us choose the controlling node $\contNode=1$ and the observing
node $\obsvNode=3$.
The transfer function is
\begin{equation}
	\tranFunFromTo{1}{3}(s) =  0.3 \frac{(s+1)^3 \prod_{i=1}^2(s^2 + \spatEigZ_i s
	+ \spatEigZ_i)}{\prod_{i=1}^5 (s^2 + \spatEig_i s + \spatEig_i)}.
\end{equation}
with $\spatEig=\{0, 0.39, 2, 2.72, 3.69\}$ and $\spatEigZ=\{0.5, 3\}$.
As indicated by (\ref{eq:tranFunProductForm}), the terms in the numerator and
the denominator products have the structure of $\vehDen(s)\contDen(s) + k
\vehNum(s)\contNum(s)$.
Moreover, since the distance between the nodes $1$ and $3$ is 2, there is also
$(s+1)^{2+1}$ in the numerator,
as follows from Corollary \ref{cor:openLoopZeros}. The weight of the path
from the node $1$ to $3$ is 0.3 (the product of the weights of the arcs).
The gains $\spatEig_i$ can be obtained as the eigenvalues of
the Laplacian matrix

{\small
\begin{equation}  
	\lapl = \left[\begin{matrix}
				1 & -1 & 0 & 0 & 0 \\
				-1 & 2 & -1 & 0 & 0 \\
				0 & -0.3 & 2.3 & -1 & -1 \\
				0 & 0 & 0 & 1 & -1 \\
				0 & 0 & -1.5 & -1 & 2.5  
			\end{matrix} \right]. 
\end{equation}}

The gains $\spatEigZ_i$ in the numerator can be obtained as
the negatives of the roots of the polynomial $\charPolNum(s) = s^2 + 3.5s +
1.5.$ Since there is only one path between $\contNode$ and $\obsvNode$, we can use
Theorem \ref{conj:onePath} to calculate the polynomial $\charPolNum(s)$. It
equals the characteristic polynomial of a matrix $\redLapl^1_{(1:3)}$, obtained
from $\lapl$ by deleting the rows and columns with indices $1,2,3$ of the
vertices on the path from 1 to 3. The polynomial is given as
\begin{equation} 
	\charPolNum(s) = \det \left(s \idMat_2 + \left[\begin{matrix}				
				1 & -1 \\
				-1 & 2.5   
			\end{matrix} \right]\right) = s^2 + 3.5s + 1.5. 
\end{equation}
As both $\spatEigZ_i$ and $\spatEig_i$ are real in this example, the poles and
zeros must lie on the root-locus curve for $\openLoop(s)=(s+1)/s^2$, as
shown in Fig.
\ref{fig:exampleGraphPoles}.
\begin{figure}  
\centering
	\includegraphics[width=0.19\textwidth]{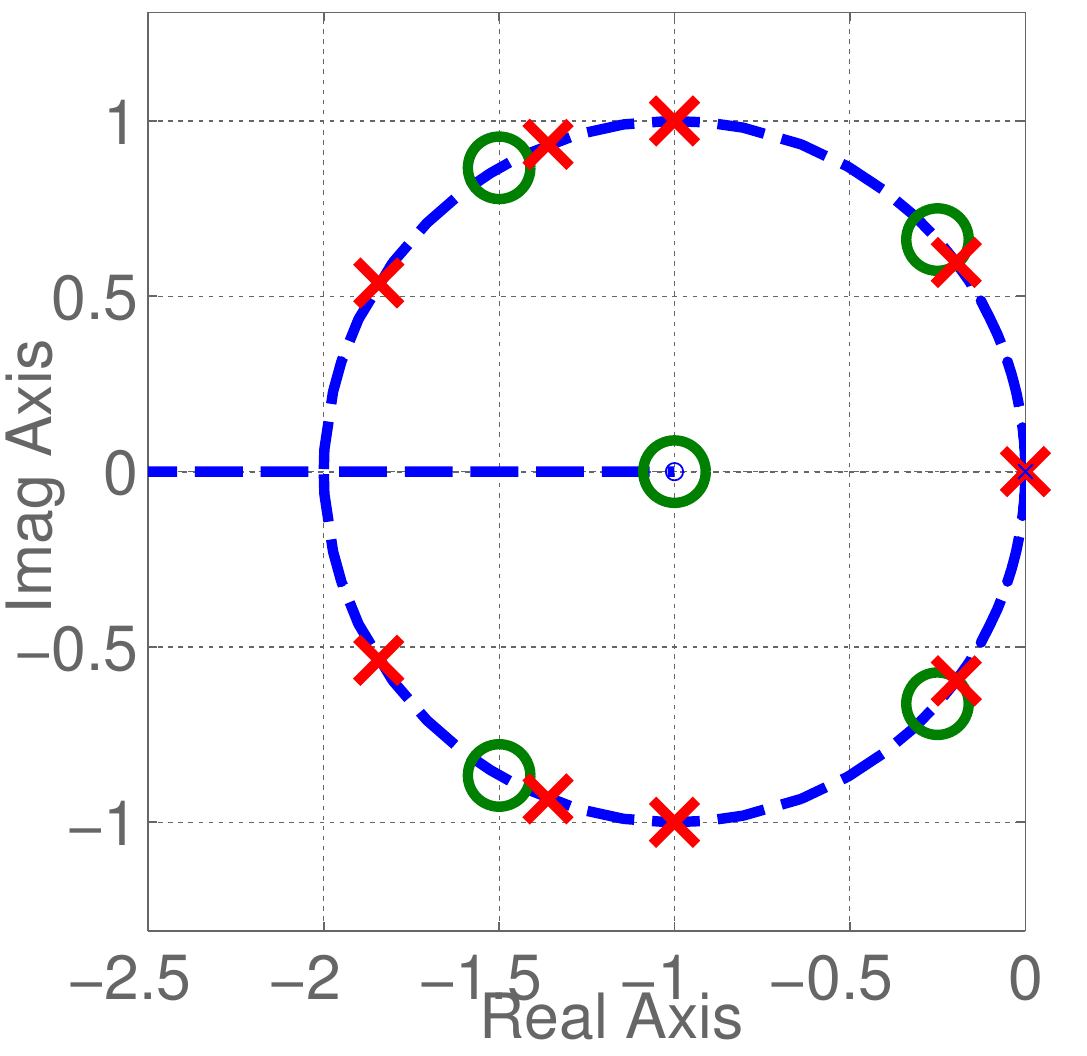}
	\caption{Poles (crosses) and zeros (circles) of $\tranFun_{1 3}(s)$ in the
	graph in Fig.
	\ref{fig:graphPath}. The root-locus curve
	for $\openLoop=(s+1)/s^2$ is dashed.
	}   
	\label{fig:exampleGraphPoles}        
\end{figure}
The minimal controllable
subspace is by Theorem \ref{thm:minContSubsp} equal to five
($\distance_{1 4}+1=4+1$), hence, the system is controllable from the node 1.
 
The transfer function $\tranFunGenInpState{1}{3}(s)$ from the input
disturbance $\inDist_1$ of agent 1 to the output $\pos_3$ is
\begin{equation}
	\tranFunDistIn{1}{3}(s) = \frac{\pos_3(s)}{\inDist_1(s)} = 0.3
	\frac{1}{s}
	\frac{(s+1)^2 \prod_{i=1}^2(s^2 + \spatEigZ_i s + \spatEigZ_i)}{\prod_{i=2}^5
	(s^2 + \spatEig_i s + \spatEig_i)}.
\end{equation}
The structure is the same as predicted in Theorem \ref{thm:genInputState}, since
$\openLoopPart(s)=\vehTf(s)=1/s$. The network part remains unchanged.
  
\section{Conclusion} 
	In this paper we considered transfer functions between two nodes in an
	arbitrary formation of identical SISO agents with an output coupling. 
Using the algebraic properties of forests in the graph, both numerator and
denominator of the transfer function were derived in a simple form of a product
of closed-loop polynomials with non-unit feedback gain. The transfer function for general input
and output consists of two parts:
the feedback part (fixed for a given pair of nodes) and the open-loop part.

The gains in the denominator and
numerator polynomials are the roots of polynomials in the single-integrator
system. If there is only one path between the controlling and observing nodes,
the numerator gains are given as eigenvalues of the principal submatrix of the
Laplacian. Finally, it is shown that the minimal dimension of the
controllable subspace grows with the maximal distance from the controlling node.

Although it is hard to tell any transient properties from the location of
poles and zeros --- there are simply too many of them --- still the product form
can serve as an analytical tool. For instance, it may help in the analysis of
the scaling in distributed control designs.
We have already applied some of the results in \cite{Herman2013b} to the
analysis of the scaling of the $\mathcal{H}_\infty$ norm, where the underlying
topology was a path graph.
	
\appendices
\section{Proof of Theorem \ref{lem:totTranFun}}
Before the proof, we need the following technical lemma.
\begin{lem}
		Let $(\lapl^k)_{\obsvNode \contNode}$ be the $\obsvNode, \contNode$ element
		of $\lapl^k$. Then
		\begin{equation}
			\sum_{i=1}^{\numVeh}\inputCoef_i v_{\obsvNode i}\spatEig_i^k =
			(\lapl^k)_{\obsvNode \contNode},
		\end{equation}
		\label{lem:eigSum}
	\end{lem}
	\begin{proof}
		Since $\inputCoef_i=\canonVect_i^T\matEigVect^{-1}\canonVect_{\contNode}$ and 
		$\eigVect_{\obsvNode
		i}=\canonVect_\obsvNode^T\matEigVect\canonVect_{i}$, we get
		\begin{IEEEeqnarray}{rCl}		
			\sum_{i=1}^{\numVeh}&&\eigVect_{\obsvNode, i}  \spatEig_i^k \inputCoef_i =
			\sum_{i=1}^{\numVeh}\canonVect_\obsvNode^T\matEigVect\canonVect_{i} 
			\spatEig_i^k \canonVect_i^T\matEigVect^{-1}\canonVect_{\contNode}
			\nonumber\\
			&&= \canonVect_\obsvNode^T\matEigVect \left(
			\sum_{i=1}^{\numVeh}\canonVect_{i} \spatEig_i^k \canonVect_i^T\right) \matEigVect^{-1}\canonVect_{\contNode}			
			=\canonVect_\obsvNode^T \matEigVect\matJ^k \matEigVect^{-1}
			\canonVect_\contNode	
			\nonumber \\	
			&&= \canonVect_\obsvNode^T \lapl^k \canonVect_\contNode =
			(\lapl^k)_{\obsvNode \contNode}.
		\end{IEEEeqnarray}
		This holds also for Jordan blocks in $\matJ$ larger than one.  
	\end{proof} 
\begin{proof}[\textbf{Proof of Theorem \ref{lem:totTranFun}}]
Let us denote the
numerator of the open loop in (\ref{eq:openLoop}) as
$\olnum(s)=\vehNum(s)\contNum(s)$ and the denominator as
$\olden(s)=\vehDen(s)\contDen(s)$. Note that the development here shows the case
with simple Jordan blocks, although the proof remains valid for the case with
larger blocks. The transfer function $\tranFunCo(s)$ can be obtained from
(\ref{eq:outputTranFunObsv}) by using a common denominator as
\begin{IEEEeqnarray}{rCl}
		\tranFunCo(s) &=& 
		\frac{\numTf(s)}{\denTf(s)}=\sum_{i=1}^{\numVeh} \inputCoef_i
		\eigVect_{\obsvNode i} \frac{\vehNum(s) \contNum(s)}{\vehDen(s)
		\contDen(s) + \spatEig_i \vehNum(s) \contNum(s)} \label{eq:prodTF} \\ \nonumber
		&=& 
		\frac{\sum_{i=1}^{\numVeh} \Big(\inputCoef_i \eigVect_{\obsvNode i}
		\, \olnum(s) \,
		\prod_{j=1, j\neq i}^{\numVeh} \left[\olden(s) + \spatEig_j
		\olnum(s) \right]\Big) } {\prod_{i=1}^{\numVeh} \left[ \olden(s) +
		\spatEig_i \olnum(s) \right]} \\
		&=& \frac{\sum_{i=1}^{\numVeh} \inputCoef_i \eigVect_{\obsvNode i} \,
		\prodPartNum_i(s) } {\prod_{i=1}^{\numVeh} \left[ \olden(s) +
		\spatEig_i \olnum(s) \right]}, \label{eq:prodPartNum}
\end{IEEEeqnarray}
with $\prodPartNum_i(s)=\olnum(s) \,
		\prod_{j=1, j\neq i}^{\numVeh} \left[\olden(s) + \spatEig_j
		\olnum(s) \right]$. Note that the polynomials in single-integrator
dynamics are $\numPol(s), \denPol(s)$, while in higher-order dynamics they are
$\numTf(s), \denTf(s)$. The denominator of (\ref{eq:prodPartNum}) is the
denominator in Theorem
\ref{lem:totTranFun}.

Having the denominator, we have to find the numerator $\numTf(s)$. The
polynomial $\prodPartNum_i(s)$ in (\ref{eq:prodPartNum}) can be expanded in terms of powers of $\olnum$ and $\olden$ as (argument $(s)$ is omitted)
\begin{IEEEeqnarray}{rCl}
		&&\prodPartNum_i = \olnum \prod_{j=1, j \neq i}^{\numVeh} [\olden + \spatEig_j
		\olnum] = \olden^{N-1} \olnum + \olden^{N-2}\olnum^2 \left[\sum_{j=1, j \neq
		 i}^{\numVeh} \spatEig_j \right]  \nonumber \\
		&&+ \olden^{N-3}\olnum^{3}\left[{\sum_{j=1, k=1,
		k \neq i \neq j}^{\numVeh}} \spatEig_j \spatEig_k \label{eq:numPos}\right]
		 + \ldots + \label{eq:expandedProdNum} \\
		&& + \olden^{1}\olnum^{N-1}\left[\sum_{j=1, j \neq i}^{\numVeh} \,\,
		\left({\prod_{{k=1, k \neq i \neq j}}^{\numVeh}} \spatEig_k\right)
		\right] + \olnum^{N}\left[\;\;\prod_{\mathclap{j=1, j \neq i}}^{\numVeh}
		\spatEig_j \, \right]. 
		\nonumber
\end{IEEEeqnarray}
Let us denote the the
coefficients at the terms $\olden^j \olnum^{\numVeh-j}$ in $\prodPartNum_i(s)$
as $\prodPartNumCoef{j}{i}$. They are given as a sum of all products of
$\numVeh-j-1$ eigenvalues.
Then the polynomial $\prodPartNum_i(s)$ can be written as
\begin{equation}
	\prodPartNum_i = \olden^{\numVeh-1} \olnum + \prodPartNumCoef{\numVeh-2}{i}
	\olden^{\numVeh-2} \olnum^2 + \ldots + \prodPartNumCoef{1}{i}
	\olden \olnum^{\numVeh-1} + \prodPartNumCoef{0}{i}
	\olnum^{\numVeh}.
\end{equation}
The coefficients $\prodPartNumCoef{j}{i}$ can be simplified. Let us start with 
\begin{equation}
 \prodPartNumCoef{N-2}{i} = \sum_{j=1, j \neq
 i}^{\numVeh} \spatEig_j = \charPolLaplCoef_{\numVeh-1} - \spatEig_i, 
 \label{eq:firstProdNumCoef}
\end{equation}
since the coefficient $\charPolLaplCoef_{\numVeh-1}$ of $\denPol(s)$ is by
(\ref{eq:charLaplCoefDef})
$\charPolLaplCoef_{\numVeh-1}=\sum_{i=1}^{\numVeh}\spatEig_i.$
Similarly, the second coefficient is using (\ref{eq:charLaplCoefDef})
\begin{equation}
\prodPartNumCoef{\numVeh-3}{i} = \; \sum_{\mathclap{j=1, k=1,
		k \neq i \neq j}}^{\numVeh}\; \spatEig_j \spatEig_k =
		\charPolLaplCoef_{\numVeh-2}-\spatEig_i(\charPolLaplCoef_{\numVeh-1}-\spatEig_i).
		 \label{eq:secondProdNumCoef}
\end{equation}
For the last coefficient we get
\begin{equation}
  \prodPartNumCoef{0}{i} = \prod_{\mathclap{j=1, j \neq i}}^{\numVeh}
		\spatEig_j = \charPolLaplCoef_1 - \spatEig_i (\charPolLaplCoef_2 - \spatEig_i
		(\charPolLaplCoef_3 - \spatEig_i(\ldots))). \label{eq:lastProdNumCoef}
\end{equation}	
Knowing the coefficients $\prodPartNumCoef{i}{j}$, the numerator polynomial
$\numTf(s)$ can be using (\ref{eq:prodPartNum}) written as
\begin{IEEEeqnarray}{rCl}
	\numTf(s) &=& \sum_{i=1}^{\numVeh} \inputCoef_i \eigVect_{\obsvNode i}
		\prodPartNum_i(s) = \olden^{\numVeh-1} \olnum \left(\sum_{i=1}^{\numVeh}
		\inputCoef_i \eigVect_{\obsvNode i}\right)
		\nonumber \\
		 &+& 
		\olden^{\numVeh-2} \olnum^2
		\left(\sum_{i=1}^{\numVeh} \inputCoef_i
		 \eigVect_{\obsvNode i}
		\, \prodPartNumCoef{\numVeh-2}{i}\right) + \ldots 
		\label{eq:numPolExpanded}
		\\ \nonumber
		&+& \olden \olnum^{\numVeh-1}
		\left( \sum_{i=1}^{\numVeh} \inputCoef_i
		 \eigVect_{\obsvNode i}
		\, \prodPartNumCoef{1}{i}\right) 
		+\olnum^{\numVeh}
		\left( \sum_{i=1}^{\numVeh} \inputCoef_i
		 \eigVect_{\obsvNode i}
		\, \prodPartNumCoef{0}{i}\right).
\end{IEEEeqnarray}
The coefficients $\charPolNumCoef_i$ of individual powers of $\olden^i
\olnum^{\numVeh-i}$ in $\numTf(s)$ can be simplified using Lemma
\ref{lem:eigSum} and the formulas for $\prodPartNumCoef{i}{j}$
(\ref{eq:firstProdNumCoef}-\ref{eq:lastProdNumCoef}). The first two read
\begin{IEEEeqnarray}{rCl}
	\charPolNumCoef_{\numVeh-1} &=& \sum_{i=1}^{\numVeh} \inputCoef_i
	\eigVect_{\obsvNode i} = (L^0)_{\obsvNode \contNode}  \label{eq:h0}\\
	\charPolNumCoef_{\numVeh-2} &=& \sum_{i=1}^{\numVeh} \inputCoef_i
		 \eigVect_{\obsvNode i}
		\, \prodPartNumCoef{\numVeh-2}{i} = \charPolLaplCoef_{\numVeh-1}
		\left(\sum_{i=1}^{\numVeh}
		\inputCoef_i
		 \eigVect_{\obsvNode i}\right)
		- \sum_{i=1}^{\numVeh}
		\inputCoef_i
		 \eigVect_{\obsvNode i} \spatEig_i 
		 \nonumber \\ &=& 
	\charPolLaplCoef_{\numVeh-1}(\lapl^0)_{\obsvNode \contNode} -
	(\lapl^1)_{\obsvNode \contNode} 
\end{IEEEeqnarray}

Using the same ideas, the other coefficients $\charPolNumCoef_i$ are
\begin{IEEEeqnarray}{rCl}
	\charPolNumCoef_{\numVeh-3} &=& \charPolLaplCoef_{N-2}(\lapl^0)_{\obsvNode
	\contNode} - \charPolLaplCoef_{N-1}(\lapl^1)_{\obsvNode \contNode} +
		(\lapl^2)_{\obsvNode \contNode} \\
	&\vdots& \nonumber \\
	\charPolNumCoef_{0} &=& \charPolLaplCoef_1 (\lapl^0)_{\obsvNode \contNode} -
	\charPolLaplCoef_2 (\lapl^1)_{\obsvNode \contNode} + \ldots +
	(\lapl^{N-1})_{\obsvNode \contNode} \label{eq:hN}.
\end{IEEEeqnarray} 
The general form is now apparent,
\begin{equation}
	\charPolNumCoef_i = \sum_{j=0}^{N-i-1} \charPolLaplCoef_{i+j+1}
	(-\lapl)_{\obsvNode\contNode}^j.
	\label{eq:higen}
\end{equation}

Using the coefficients $\charPolNumCoef_i$ in (\ref{eq:h0})-(\ref{eq:hN}), the
numerator $\numTf(s)$ in (\ref{eq:numPolExpanded}) equals
\begin{IEEEeqnarray}{rCl}
		\numTf(s) &=& \olnum(s)\bigg(\charPolNumCoef_{\numVeh-1} \olden(s)^{\numVeh-1}
		+ \charPolNumCoef_{\numVeh-2} \olden^{\numVeh-2}(s)\olnum(s) \nonumber \\
		&&+ \charPolNumCoef_{\numVeh-3} \olden^{\numVeh-3}(s)\olnum^2(s) + \ldots +
		\charPolNumCoef_{0}\olnum^{N-1}(s)\bigg). \label{eq:numeratorComplete}
\end{IEEEeqnarray} 

Now we show that the coefficients $\charPolNumCoef_i$ in (\ref{eq:numeratorComplete}) are equal to
	the coefficients $\charPolNumCoefSingle_i$ of the numerator polynomial
	$\numPol(s)$ in the single integrator dynamics, i.e., $\charPolNumCoef_i=\charPolNumCoefSingle_i, \forall \, i$.

	To see this, Corollary 4 in \cite{Chebotarev2002} gives us a relation
	\begin{equation}
		\text{adj}(s \idMat + \lapl) = \sum_{k=0}^{N-1}\left( \sum_{j=0}^{\numVeh-k-1}
		\charPolLaplCoef_i s^{\numVeh-j-1}  \right) (-\lapl^{k}/s^k).
		\label{eq:adjMatGi}
	\end{equation} 
	The coefficient matrix $\Gamma_{i}$ at $s^i$ in (\ref{eq:adjMatGi}) is then
	defined as
	\begin{equation}
	\Gamma_i = \sum_{j=0}^{\numVeh - i - 1} \charPolLaplCoef_{i+j+1} (-\lapl)^j. 
	\end{equation}
	Taking as an element of interest the $\obsvNode, \contNode$th element in
	$\text{adj}(s \idMat + \lapl)$, we see by (\ref{eq:higen}) that the
	coefficients $(\Gamma_i)_{\obsvNode \contNode }=\charPolNumCoef_i$. Moreover,
	since by (\ref{eq:numPolAsAdjugateMat}) $\text{adj}(s \idMat +
	\lapl)_{\obsvNode \contNode}$ is equal to the numerator polynomial in single-integrator dynamics, we get
	$\charPolNumCoefSingle_i=\charPolNumCoef_i, \; \forall i$.

All the coefficients $\charPolNumCoefSingle_i$ are functions of the powers of
the Laplacian.
Using Lemma \ref{lem:numwalks}, it is clear that $\charPolNumCoefSingle_i = 0$
for $i > \numVeh-\distanceCO$, since all
$(\lapl^{j})_{\obsvNode \contNode}$ for $j=0, 1, \ldots,
\distanceCO-1$ are zeros.
Then in (\ref{eq:numPolGamma}), $\orderQ=\numVeh-\distanceCO-1$ also for directed weighted graphs.
This result allows us to rewrite (\ref{eq:numeratorComplete}) as
\begin{IEEEeqnarray}{rCl}
		\numTf&&(s) = \olnum^{1+\distanceCO}(s)\Big(
		\charPolNumCoefSingle_{\numVeh-\distanceCO-1}
		\olden^{\numVeh-1-\distanceCO}(s) \label{eq:factoredNum} \\ \nonumber
		&&+
		\charPolNumCoefSingle_{\numVeh-\distanceCO-2}\olden^{\numVeh-2-\distanceCO}(s)\olnum(s)
		+ \ldots + \charPolNumCoefSingle_{0}\olnum^{N-1-\distanceCO}(s) \Big).		
	\end{IEEEeqnarray}

Previous equation can be factored into a product
\begin{equation}
	\numTf(s) =
	\charPolNumCoefSingle_{\numVeh-\distanceCO-1} \olnum^{1+\distanceCO}(s) \:
	\prod_{i=1}^{\mathclap{\numVeh-1-\distanceCO}} \;\; \Big(\olden(s) +
	\spatEigZ_i \olnum(s)\Big), \label{eq:numProd}
\end{equation}
where the scalars $-\spatEigZ_i$ are the roots of the polynomial $\numPol(s)$
defined in (\ref{eq:numPolGamma}). They are thus the zeros of the transfer function for the single
integrator dynamics.

Note that
$\charPolNumCoef_{\numVeh-\distanceCO-1}=\charPolNumCoefSingle_{\numVeh-\distanceCO-1}=\pathWeightCo$
by Lemma \ref{lem:coefForestWeights}. Then we get the numerator as
\begin{equation}
	\numTf(s) = \pathWeightCo \: \olnum^{1+\distanceCO}(s) \,
	\prod_{i=1}^{\mathclap{\numVeh-1-\distanceCO}} \;\; \Big(\vehDen(s)\contDen(s)
	+ \spatEigZ_i \vehNum(s)\contNum(s)\Big). \label{eq:finalNumerator}
\end{equation}
Now in (\ref{eq:finalNumerator}) and (\ref{eq:prodPartNum}) we have both the
numerator $\numTf(s)$ and the denominator $\denTf(s)$ of (\ref{eq:tranFunProductForm}), which concludes the proof.
\end{proof}
 
\bibliographystyle{ieeetran} 
\bibliography{Papers-TransferFunctionsInGraphs}
\end{document}